\documentclass[%
 reprint,
 superscriptaddress,
 amsmath,amssymb,
 aps,
 amsthm,
 pra,
]{revtex4-1}

\usepackage{graphicx}
\usepackage{dcolumn}
\usepackage{bm}
\usepackage{braket}%
\usepackage{theorem}
\usepackage{multirow}
\usepackage{framed}
\usepackage{enumerate}

\usepackage{color}

\theorembodyfont{\upshape}

\definecolor{gray}{RGB}{128,128,128}

\newtheorem{thm}{Theorem}

\newtheorem{lemma}[thm]{Lemma}

\newtheorem{proof}{Proof}

\newcommand{\hA}{\hat{A}}
\newcommand{\hB}{\hat{B}}

\newcommand{\hG}{\hat{G}}

\newcommand{\hP}{\hat{P}}

\newcommand{\hU}{\hat{U}}

\newcommand{\hW}{\hat{W}}
\newcommand{\hX}{\hat{X}}
\newcommand{\hY}{\hat{Y}}

\newcommand{\ha}{\hat{a}}

\newcommand{\hc}{\hat{c}}

\newcommand{\hz}{\hat{z}}

\newcommand{\hPi}{\hat{\Pi}}
\newcommand{\hPhi}{\hat{\Phi}}

\newcommand{\hrho}{\hat{\rho}}

\newcommand{\hsigma}{\hat{\sigma}}

\newcommand{\mG}{\mathcal{G}}
\newcommand{\mH}{\mathcal{H}}
\newcommand{\mI}{\mathcal{I}}

\newcommand{\mM}{\mathcal{M}}

\newcommand{\mP}{\mathcal{P}}

\newcommand{\mR}{\mathcal{R}}
\newcommand{\mS}{\mathcal{S}}

\newcommand{\mX}{\mathcal{X}}

\newcommand{\ident}{\hat{1}}

\newcommand{\Real}{\mathbf{R}}

\newcommand{\POVM}{\mM}

\newcommand{\QED}{\hspace*{0pt}\hfill $\blacksquare$}

\newcommand{\Tr}{{\rm Tr}}

\def\gauss_sym#1{{\lfloor #1 \rfloor}}

\renewcommand{\Real}{{\mR}}
\newcommand{\fmin}{f_{\min}}
\newcommand{\Prob}{\mP}
\newcommand{\mm}{\circ}
\renewcommand{\c}{\circ}
\newcommand{\g}{{(g)}}
\newcommand{\inv}[1]{\bar{#1}}

\begin{document}

\preprint{APS/123-QED}

\title{Generalized quantum state discrimination problems}

\affiliation{%
 Yokohama Research Laboratory, Hitachi, Ltd.,
 Yokohama, Kanagawa 244-0817, Japan
}%
\affiliation{
 Quantum Communication Research Center, Quantum ICT Research Institute, Tamagawa University,
 Machida, Tokyo 194-8610, Japan
}%
\affiliation{
 School of Information Science and Technology,
 Aichi Prefectural University,
 Nagakute, Aichi 480-1198, Japan
}%
\affiliation{%
 Quantum Information Science Research Center, Quantum ICT Research Institute,
 Tamagawa University, Machida, Tokyo 194-8610, Japan
}%

\author{Kenji Nakahira}
\affiliation{%
 Yokohama Research Laboratory, Hitachi, Ltd.,
 Yokohama, Kanagawa 244-0817, Japan
}%
\affiliation{%
 Quantum Information Science Research Center, Quantum ICT Research Institute,
 Tamagawa University, Machida, Tokyo 194-8610, Japan
}%

\author{Kentaro Kato}
\affiliation{
 Quantum Communication Research Center, Quantum ICT Research Institute, Tamagawa University,
 Machida, Tokyo 194-8610, Japan
}%

\author{Tsuyoshi \surname{Sasaki Usuda}}
\affiliation{
 School of Information Science and Technology,
 Aichi Prefectural University,
 Nagakute, Aichi 480-1198, Japan
}%
\affiliation{%
 Quantum Information Science Research Center, Quantum ICT Research Institute,
 Tamagawa University, Machida, Tokyo 194-8610, Japan
}%

\date{\today}

\begin{abstract}
 We address a broad class of optimization problems of finding quantum measurements,
 which includes the problems of finding an optimal measurement in the Bayes criterion
 and a measurement maximizing the average success probability
 with a fixed rate of inconclusive results.
 Our approach can deal with any problem in which each of the objective and constraint functions is formulated by
 the sum of the traces of the multiplication of a Hermitian operator and a detection operator.
 We first derive dual problems and necessary and sufficient conditions for an optimal measurement.
 We also consider the minimax version of these problems
 and provide necessary and sufficient conditions for a minimax solution.
 Finally, for optimization problem having a certain symmetry,
 there exists an optimal solution with the same symmetry.
 Examples are shown to illustrate how our results can be used.
\end{abstract}

\pacs{03.67.Hk}
\maketitle


\section{Introduction}

Discrimination between quantum states is a fundamental topic in quantum information theory.
Research in quantum state discrimination was pioneered by
Helstrom, Holevo, and Yuen {\it et al.} \cite{Hel-1976,Hol-1973,Yue-Ken-Lax-1975}
in the 1970s and has attracted intensive attention.
It is well known in quantum mechanics that nonorthogonal states cannot be
discriminated with certainty.
Thus, optimal measurement strategies have been proposed under various criteria.
Among them, one of the most widely investigated is the Bayes criterion,
or the criterion of minimum average error probability \cite{Hol-1973,Yue-Ken-Lax-1975,Hel-1976}.
In the Bayes criterion, necessary and sufficient conditions for obtaining an optimal measurement
have been formulated \cite{Hol-1973,Yue-Ken-Lax-1975,Hel-1976,Eld-Meg-Ver-2003},
and closed-form analytical expressions for optimal measurements have also been derived
in some cases (see e.g.,
\cite{Bel-1975,Ban-Kur-Mom-Hir-1997,Usu-Tak-Hat-Hir-1999,Eld-For-2001,Nak-Usu-2012-Bayes}).
This criterion is based on the assumption that prior probabilities of the states are known.
In contrast, if these prior probabilities are unknown,
then minimax criteria are often used \cite{Hir-Ike-1982,Osa-Ban-Hir-1996}.
Necessary and sufficient conditions for a measurement minimizing
the worst case of the average error probability
in the minimax strategy have been found \cite{Hir-Ike-1982}.
This result has also been extended to the average Bayes cost \cite{Kat-2012},

Other types of optimal measurements have been investigated.
In the case in which prior probabilities of the states are known,
an example concerns a measurement that achieves low average error probability
at the expense of allowing for
a certain fraction of inconclusive results \cite{Iva-1987,Die-1988,Per-1988}.
In particular, an unambiguous (or error-free) measurement that maximizes the success probability,
which is called an optimal unambiguous measurement, has been well studied \cite{Iva-1987,Die-1988,Per-1988}.
A measurement that maximizes the average success probability
with a fixed average failure (or inconclusive) probability,
which is called an optimal inconclusive measurement,
has also been studied \cite{Che-Bar-1998-inc,Eld-2003-inc,Fiu-Jez-2003}.
Moreover, a measurement that maximizes the average success probability
where a certain fixed average error probability is allowed,
which we call an optimal error margin measurement, has also been investigated
\cite{Tou-Ada-Ste-2007,Hay-Has-Hor-2008,Sug-Has-Hor-Hay-2009}.
On the other hand, in the case in which prior probabilities are unknown,
several types of measurements based on the minimax strategy have been proposed,
such as a measurement that minimizes the maximum probability of detection errors \cite{Dar-Sac-Kah-2005}
and a measurement with a certain fraction of inconclusive results \cite{Nak-Kat-Usu-2013-minimax}.
Properties of optimal measurements in the above criteria,
such as necessary and sufficient conditions for optimal solutions,
have been derived for each criterion.

In this paper, we investigate optimization problems of finding optimal quantum measurements
and their minimax versions
that are applicable to a wide range of quantum state discrimination problems.
Our approach can deal with any problem in which each of the objective and constraint functions is formulated by
the sum of the traces of the multiplication of a Hermitian operator and a detection operator,
which implies that any problems related to finding any of the optimal measurements
described above can be formulated as our problems.
Thus, we can say that our approach can provide a unified treatment in a large class of problems.
The results obtained in this paper would be valuable from the practical point of view;
for example, they not only provide a broader perspective than the results for a particular problem,
but also can apply to many problems that have not been reported previously,
some examples of which are presented in this paper.
To obtain knowledge about an optimal measurement in a new criterion
has the potential to create a new application of quantum state discrimination.

In Sec.~\ref{sec:gen_opt}, we provide a generalized optimization problem
in which each of the objective and constraint functions is formulated by
the sum of the traces of the multiplication of a Hermitian operator and a detection operator.
We derive its dual problem and necessary and sufficient conditions for an optimal measurement.
In Sec.~\ref{sec:minimax}, we discuss the minimax version of our generalized problem
and provide necessary and sufficient conditions for a minimax solution.
In Sec.~\ref{sec:sym}, we demonstrate that if a given problem has a certain symmetry,
then there exists an optimal solution with the same symmetry.
Finally, we present some examples to illustrate the applicability of our results in Sec.~\ref{sec:example}.

\section{Generalized optimal measurement} \label{sec:gen_opt}

\subsection{Formulation}

We consider a quantum measurement on a Hilbert space $\mH$.
Such a quantum measurement can be modeled by
a positive operator-valued measure (POVM)
$\Pi = \{ \hPi_m : m \in \mI_M \}$ on $\mH$,
where $M$ is the number of the detection operators
and $\mI_k = \{ 0, 1, \cdots, k-1 \}$.
An example of a quantum measurement is an optimal measurement
for distinguishing $R$ quantum states represented by density operators $\hrho_r ~(r \in \mI_R)$.
The density operator $\hrho_r$ satisfies $\hrho_r \ge 0$
and has unit trace, i.e., $\Tr ~\hrho_r = 1$,
where $\hA \ge 0$ denotes that $\hA$ are positive semidefinite
(similarly, $\hA \ge \hB$ denotes $\hA - \hB \ge 0$).
A minimum error measurement is such an optimal measurement,
which can be expressed by a POVM with $M = R$ detection operators.
A quantum measurement that may return an inconclusive answer
can be expressed by a POVM with $M = R + 1$ detection operators;
in this case the detection operator $\hPi_r$ $~(r \in \mI_R)$ corresponds to identification of
the state $\hrho_r$, while $\hPi_R$ corresponds to the inconclusive answer.

Let $\POVM$ be the entire set of POVMs on $\mH$ that consist of $M$ detection operators.
$\Pi \in \POVM$ satisfies
\begin{eqnarray}
 \hPi_m &\ge& 0, ~~~ \forall m \in \mI_M, \nonumber \\
 \sum_{m=0}^{M-1} \hPi_m &=& \ident, \label{eq:POVM}
\end{eqnarray}
where $\ident$ is the identity operator on $\mH$.
In addition, let $\mS$ and $\mS_+$ be the entire sets of Hermitian operators on $\mH$
and semidefinite positive operators on $\mH$, respectively.
Let $\Real$ and $\Real_+$ be the entire sets of real numbers and nonnegative real numbers, respectively,
and $\Real_+^N$ be the entire set of collections of $N$ nonnegative real numbers.

Here, we consider a generalized optimization problem.
The conditional probability that the measurement outcome is $m$ when
a quantum state $\hrho$ is given is represented by $\Tr(\hrho\hPi_m)$,
and thus there exist many optimization problems of finding optimal quantum measurements
such that each of the objective and constraint functions is expressed by
a linear combination of forms $\Tr(\hrho_r\hPi_m)$.
For this reason, we consider the following optimization problem:
\begin{eqnarray}
 \begin{array}{ll}
  {\rm maximize} & \displaystyle f(\Pi) = \sum_{m=0}^{M-1} \Tr(\hc_m \hPi_m) \\
  {\rm subject~to} & \Pi \in \POVM^\mm, \\
 \end{array} \label{eq:primal}
\end{eqnarray}
where $\hc_m \in \mS$ holds for any $m \in \mI_M$.
(Note that any linear combination of positive semidefinite operators is a Hermitian operator.)
$\POVM^\mm$ is expressed by
\begin{eqnarray}
 \hspace{-1em}
  \POVM^\mm &=& \left\{ \Pi \in \POVM : \sum_{m=0}^{M-1} \Tr(\ha_{j,m} \hPi_m) \le b_j, ~ \forall j \in \mI_J \right\},
  \label{eq:POVMmm}
\end{eqnarray}
where $\ha_{j,m} \in \mS$ and $b_j \in \Real$ hold for any $m \in \mI_M$ and $j \in \mI_J$.
$J$ is a nonnegative integer.
We should mention that an equality constraint (e.g., $\Tr(\ha_{j,0} \hPi_0) = b_j$) can be replaced
by two inequality constraints (e.g., $\Tr(\ha_{j,0} \hPi_0) \le b_j$ and $\Tr(-\ha_{j,0} \hPi_0) \le -b_j$).
We call an optimal solution to problem (\ref{eq:primal})
a generalized optimal measurement or simply an optimal measurement.
Problem (\ref{eq:primal}) is said to be a primal problem.
Since $f(\Pi)$ is linear in $\Pi$ and $\POVM^\mm$ is convex,
problem (\ref{eq:primal}) is a convex optimization problem.
Note that since the constraint of $\Pi \in \POVM$, i.e., Eq.~(\ref{eq:POVM}),
can be rewritten as $\Tr(\hrho\hPi_m) \ge 0$
and $\sum_{m=0}^{M-1} \Tr(\hrho\hPi_m) = 1$ for any density operator $\hrho$,
we can say that each of the objective and constraint functions is formulated by
the sum of the traces of the multiplication of a Hermitian operator and a detection operator.

\subsection{Examples} \label{subsec:example}

We give some examples of optimization problems of finding quantum measurements
that can be formulated as problem (\ref{eq:primal}).
Let us consider discrimination between $R$ quantum states $\{ \hrho_r : r \in \mI_R \}$
with prior probabilities $\{ \xi_r : r \in \mI_R \}$.

\subsubsection{Optimal measurement in the Bayes criterion} \label{subsubsec:bayes}

The optimization problem of finding an optimal measurement in the Bayes criterion
is formulated as \cite{Hol-1973,Yue-Ken-Lax-1975,Hel-1976}
\begin{eqnarray}
 \begin{array}{ll}
  {\rm minimize} & \displaystyle \sum_{m=0}^{R-1} \Tr(\hW_m \hPi_m) \\
  {\rm subject~to} & \displaystyle \Pi \in \POVM. \\
 \end{array} \label{eq:bayes_primal}
\end{eqnarray}
$\hW_m \in \mS_+$ $~(m \in \mI_R)$ can be expressed by
\begin{eqnarray}
 \hW_m &=& \sum_{r=0}^{R-1} \xi_r B_{m,r} \hrho_r,
\end{eqnarray}
where $B_{m,r} \in \Real_+$ holds for any $m,r \in \mI_R$.
This problem can be written as the form of problem (\ref{eq:primal}) with
\begin{eqnarray}
 M &=& R, \nonumber \\
 J &=& 0, \nonumber \\
 \hc_m &=& - \hW_m.
\end{eqnarray}

\subsubsection{Optimal error margin measurement} \label{subsubsec:error_margin}

An optimal error margin measurement is a measurement maximizing the average success
probability under the constraint that the average error probability
is not greater than a given value $\varepsilon$ with $0 \le \varepsilon \le 1$
\cite{Tou-Ada-Ste-2007,Hay-Has-Hor-2008,Sug-Has-Hor-Hay-2009}.
In particular, if $\varepsilon = 0$, then
an optimal error margin measurement is equivalent to an optimal unambiguous measurement.
The optimization problem of finding an optimal error margin measurement is formulated as
\begin{eqnarray}
 \begin{array}{ll}
  {\rm maximize} & \displaystyle \sum_{r=0}^{R-1} \xi_r \Tr(\hrho_r \hPi_r) \\
  {\rm subject~to} & \displaystyle \Pi \in \POVM, ~ \sum_{r=0}^{R-1} \xi_r \Tr[\hrho_r (\hPi_r + \hPi_R)] \ge 1 - \varepsilon, \\
 \end{array} \label{eq:err_primal}
\end{eqnarray}
where we consider that the statement that the average error probability is not greater than $\varepsilon$
is equivalent to the statement that the sum of the average success and failure probabilities
is not less than $1 - \varepsilon$.
This problem can be written as the form of problem (\ref{eq:primal}) with
\begin{eqnarray}
 M &=& R + 1, \nonumber \\
 J &=& 1, \nonumber \\
 \hc_m &=&
  \left\{
   \begin{array}{ll}
	\xi_m \hrho_m, & ~ m < R, \\
	0, & ~ m = R,
   \end{array} \right. \nonumber \\
 \ha_{0,m} &=&
  \left\{
   \begin{array}{ll}
	- \xi_m \hrho_m, & ~ m < R, \\
	- \hG, & ~ m = R,
   \end{array} \right. \nonumber \\
 b_0 &=& \varepsilon - 1,
\end{eqnarray}
where
\begin{eqnarray}
 \hG &=& \sum_{r=0}^{R-1} \xi_r \hrho_r. \label{eq:G}
\end{eqnarray}
Note that an optimal error margin measurement has strong relationship
with an optimal inconclusive measurement \cite{Nak-Usu-Kat-2012-GUInc,Her-2012}.
However, if one wants to obtain an optimal error margin measurement for a given $\varepsilon$,
then one needs to solve problem (\ref{eq:err_primal})
instead of the problem of finding an optimal inconclusive measurement.

\subsubsection{Optimal inconclusive measurement with a lower bound on success probabilities}
 \label{subsubsec:bounded_inc}

Another example is an extension of the problem of finding an optimal inconclusive measurement.
An optimal inconclusive measurement is a measurement maximizing the average success
probability under the constraint that the average failure probability equals a given value $p$
with $0 \le p \le 1$
\cite{Che-Bar-1998-inc,Eld-2003-inc,Fiu-Jez-2003}.
Here, we add the constraint that for each $r \in \mI_R$
the success probability of the state $\hrho_r$, i.e., $\Tr(\hrho_r \hPi_r)$,
is not less than a given value $q$ with $0 \le q \le 1$.
When $q = 0$, an optimal solution is an optimal inconclusive measurement.
This problem is formulated as
\begin{eqnarray}
 \begin{array}{ll}
  {\rm maximize} & \displaystyle \sum_{r=0}^{R-1} \xi_r \Tr(\hrho_r \hPi_r) \\
  {\rm subject~to} & \Pi \in \POVM, ~ \Tr(\hrho_r \hPi_r) \ge q, ~ \forall r \in \mI_R, \\
  & \Tr(\hG\hPi_R) = p, \\
 \end{array} \label{eq:bounded_inc_primal}
\end{eqnarray}
where $\hG$ is defined by Eq.~(\ref{eq:G}).
Since the optimal value of problem (\ref{eq:bounded_inc_primal}) is monotonically decreasing
with respect to $p$, we obtain the same solution if the last constraint
of problem (\ref{eq:bounded_inc_primal}) is replaced with $\Tr(\hG \hPi_R) \ge p$.
Thus, this problem is equivalent to problem (\ref{eq:primal}) with
\begin{eqnarray}
 M &=& J = R + 1, \nonumber \\
 \hc_m &=&
  \left\{
   \begin{array}{ll}
	\xi_m \hrho_m, & ~ m < R, \\
	0, & ~ m = R,
   \end{array} \right. \nonumber \\
 \ha_{j,m} &=&
  \left\{
   \begin{array}{ll}
	- \delta_{m,j} \hrho_m, & ~ j < R, \\
	- \delta_{m,R} \hG, & ~ j = R, \\
   \end{array} \right. \nonumber \\
 b_j &=&
  \left\{
   \begin{array}{ll}
	- q, & ~ j < R, \\
	- p, & ~ j = R, \\
   \end{array} \right. \label{eq:bounded_inc_abc}
\end{eqnarray}
where $\delta_{k,k'}$ is the Kronecker delta.
Note that if $q > q'$ holds, where $q'$ is the average success probability of
an optimal inconclusive measurement with the average failure probability of $p$,
then this problem is infeasible, i.e., $\POVM^\mm$ is empty.
We discuss this problem in detail in Subsec.~\ref{subsec:example_bounded_inc}.

\subsection{Dual problem} \label{subsec:dual}

In this subsection, we show the dual problem of problem (\ref{eq:primal}).
We also show that the optimal values of primal problem (\ref{eq:primal}) and
the dual problem are the same.

\begin{thm} \label{thm:dual}
 Let us consider problem (\ref{eq:primal}).
 We also consider the following optimization problem
 \begin{eqnarray}
  \begin{array}{ll}
   {\rm minimize} & \displaystyle s(\hX, \lambda) = \Tr~\hX + \sum_{j=0}^{J-1} \lambda_j b_j \\
   {\rm subject~to} & \hX \ge \hz_m(\lambda), ~~ \forall m \in \mI_M  \\
  \end{array} \label{eq:dual}
 \end{eqnarray}
 with variables $\hX \in \mS$ and $\lambda = \{ \lambda_j \in \Real_+ : j \in \mI_J \} \in \Real_+^J$,
 where
 \begin{eqnarray}
  \hz_m(\lambda) &=& \hc_m - \sum_{j=0}^{J-1} \lambda_j \ha_{j,m}. \label{eq:zm}
 \end{eqnarray}
 If $\POVM^\mm$ is not empty, then
 the optimal values of problems (\ref{eq:primal}) and (\ref{eq:dual}) are the same.
\end{thm}

Problem (\ref{eq:dual}) is called the dual problem of problem (\ref{eq:primal}).
Note that in general $\hX$ satisfying the constraints of problem (\ref{eq:dual}) is not in $\mS_+$;
however, it is obvious that if $m \in \mI_M$ exists such that $\hz_m \in \mS_+$, then $\hX \in \mS_+$ holds.

\begin{proof}
 Let us define the function $L$ as
 \begin{eqnarray}
  \lefteqn{ L(\Pi, \sigma, \hX, \lambda) } \nonumber \\
  &=& f(\Pi) + \sum_{m=0}^{M-1} \Tr(\hsigma_m \hPi_m)
  + \Tr \left[\hX \left( \ident - \sum_{m=0}^{M-1} \hPi_m \right) \right] \nonumber \\
  & & \mbox{} + \sum_{j=0}^{J-1} \lambda_j \left[ b_j - \sum_{m=0}^{M-1} \Tr(\ha_{j,m} \hPi_m) \right], \label{eq:L}
 \end{eqnarray}
 where $\sigma = \{ \hsigma_m \in \mS_+ : m \in \mI_M \}$, $\hX \in \mS$, and $\lambda \in \Real_+^J$.
 Note that $L$ is called the Lagrangian for problem (\ref{eq:primal}).
 Substituting Eqs.~(\ref{eq:primal}),(\ref{eq:dual}), and (\ref{eq:zm}) into Eq.~(\ref{eq:L})
 gives
 \begin{eqnarray}
  \lefteqn{ L(\Pi, \sigma, \hX, \lambda) } \nonumber \\
  &=& s(\hX, \lambda) 
   + \sum_{m=0}^{M-1} \Tr[(\hsigma_m + \hz_m(\lambda) - \hX) \hPi_m]. \label{eq:L2}
 \end{eqnarray}
 Let us consider the following optimization problem:
 \begin{eqnarray}
  \begin{array}{ll}
   {\rm minimize} & s_\sigma(\sigma, \hX, \lambda) \\
   {\rm subject~to} & \hsigma_m \in \mS_+, ~~ \forall m \in \mI_M, \\
   ~                & \hX \in \mS, \\
   ~                & \lambda \in \Real_+^J, \\
  \end{array} \label{eq:dual_s}
 \end{eqnarray}
 where
 \begin{eqnarray}
  s_\sigma(\sigma, \hX, \lambda) &=& \max_{\Pi \in \mS_+^M} L(\Pi, \sigma, \hX, \lambda) \label{eq:gs}
 \end{eqnarray}
 and $\mS_+^M = \{ \hPi_m \in \mS_+ : m \in \mI_M \}$.
 Let $\mX = \{ \hX : \hX \ge \hsigma_m + \hz_m(\lambda), ~ \forall m \in \mI_M \}$.
 The second term of the right-hand side of Eq.~(\ref{eq:L2}) is nonpositive if $\hX \in \mX$,
 and can be infinite if $\hX \not\in \mX$.
 Therefore, from Eq.~(\ref{eq:gs}), $s_\sigma(\sigma, \hX, \lambda)$ can be expressed as
 \begin{eqnarray}
  s_\sigma(\sigma, \hX, \lambda) &=&
   \left\{
	\begin{array}{ll}
	 \displaystyle s(\hX, \lambda), & ~ \hX \in \mX, \\
	 \infty, & ~ {\rm otherwise}. \\
	\end{array} \right. \label{eq:gs2}
 \end{eqnarray}
 From Eq.~(\ref{eq:gs2}), it follows that there exists an optimal solution to problem (\ref{eq:dual_s})
 such that $\hsigma_m = 0$ holds for any $m \in \mI_M$.
 Indeed, if $(\sigma, \hX, \lambda)$ is an optimal solution to problem (\ref{eq:dual_s})
 (in this case, $\hX \in \mX$ holds from Eq.~(\ref{eq:gs2})),
 then $(\{ \hsigma'_m = 0 : m \in \mI_M \}, \hX, \lambda)$ is also an optimal solution.
 Hence, problem (\ref{eq:dual_s}) can be rewritten by problem (\ref{eq:dual}).

 Slater's condition is known to a sufficient condition under which, if the primal problem is convex,
 the optimal values of the primal and dual problems are the same \cite{Sla-1950}.
 Since each constraint of primal problem (\ref{eq:primal}),
 including the constraint of $\Pi \in \POVM$, is expressed as
 a form of $u_j(\Pi) \le 0$, where $u_j$ is an affine function of $\Pi$,
 from Ref.~\cite{Boy-2009},
 (the refined form of) Slater's condition is that the primal problem is feasible, i.e., $\POVM^\mm$ is not empty.
 Thus, since Slater's condition holds,
 the optimal values of problems (\ref{eq:primal}) and (\ref{eq:dual}) are the same.
 \QED
\end{proof}

It is worth noting that some attempts have been made to obtain the maximum average success probability
without using the fact that POVMs describe quantum measurements
\cite{Kim-Miy-Ima-2009,Bae-Hwa-Han-2011,Bae-2013}.
In Ref.~\cite{Kim-Miy-Ima-2009}, the dual problem to the problem of
finding a minimum error measurement was derived from general probabilistic theories.
In Refs.~\cite{Bae-Hwa-Han-2011,Bae-2013}, 
the dual problem was derived from ``ensemble steering,'' which determines
what states one party can prepare on the other party's system by sharing a bipartite state.
In the same way, we can derive dual problem (\ref{eq:dual})
without using POVMs (see Appendix~\ref{append:nosig}).
However, it might not be easy to prove that the optimal value of problem (\ref{eq:dual}) is attained
by using these approaches.

\subsection{Conditions for an optimal measurement} \label{subsec:opt_condition}

Necessary and sufficient conditions for an optimal measurement in several problems
(such as a minimum error measurement and an optimal inconclusive measurement)
have been derived
\cite{Hol-1973,Hel-1976,Yue-Ken-Lax-1975,Eld-Meg-Ver-2003,Fiu-Jez-2003,Eld-2003-inc,Eld-2003-unamb,Nak-Usu-2013-group}.
The following theorem extends these results to our more general setting.

\begin{thm} \label{thm:condition}
 Suppose that a POVM $\Pi$ is in $\POVM^\mm$.
 The following statements are all equivalent.
 \begin{enumerate}[(1)]
  \setlength{\parskip}{0cm}
  \setlength{\itemsep}{0cm}
  \item $\Pi$ is an optimal measurement of problem (\ref{eq:primal}).
  \item $\hX \in \mS$ and $\lambda \in \Real_+^J$ exist such that
		\begin{eqnarray}
		 \hX - \hz_m(\lambda) &\ge& 0, ~ \forall m \in \mI_M, \label{eq:cond_Xz} \\
		 \lbrack \hX - \hz_m(\lambda) \rbrack \hPi_m &=& 0, ~ \forall m \in \mI_M, \label{eq:cond_XzPi} \\
		 \lambda_j \left[ b_j - \sum_{m=0}^{M-1} \Tr(\ha_{j,m} \hPi_m) \right] &=& 0,
		  ~ \forall j \in \mI_J. \label{eq:cond_b}
		\end{eqnarray}
  \item $\lambda \in \Real_+^J$ exists such that
		\begin{eqnarray}
		 \sum_{n=0}^{M-1} \hz_n(\lambda) \hPi_n - \hz_m(\lambda) &\ge& 0,
		  ~ \forall m \in \mI_M, \label{eq:cond_Xlambda} \\
		 \lambda_j \left[ b_j - \sum_{m=0}^{M-1} \Tr(\ha_{j,m} \hPi_m) \right] &=& 0,
		  ~ \forall j \in \mI_J. \label{eq:cond_b2}
		\end{eqnarray}
 \end{enumerate}
\end{thm}

Note that, from Eq.~(\ref{eq:cond_XzPi}),
the support of the detection operator $\hPi_m$ of the optimal measurement
is included in the kernel of $\hX - \hz_m(\lambda)$ for any $m \in \mI_M$.

\begin{proof}
 It is sufficient to show (1) $\Rightarrow$ (2), (2) $\Rightarrow$ (3), and (3) $\Rightarrow$ (1).

 First, we show (1) $\Rightarrow$ (2).
 Suppose that $(\hX, \lambda)$ is an optimal solution to dual problem (\ref{eq:dual}).
 Let $\hsigma_m = 0$ for any $m \in \mI_M$.
 It is obvious from Eq.~(\ref{eq:dual}) that Eq.~(\ref{eq:cond_Xz}) holds.
 From Theorem~\ref{thm:dual}, $f(\Pi) = s(\hX, \lambda)$ holds.
 Moreover, from $\Pi \in \POVM^\mm$,
 the second and third terms of the right-hand side of Eq.~(\ref{eq:L}) are zero,
 and the fourth term is nonnegative,
 which yields
 \begin{eqnarray}
  L(\Pi, \sigma, \hX, \lambda) \ge f(\Pi) = s(\hX, \lambda). \label{eq:L_ge_g}
 \end{eqnarray}
 In contrast, from Eq.~(\ref{eq:cond_Xz}) and the fact that 
 the trace of the multiplication of two positive semidefinite operators is nonnegative,
 $\Tr[(\hX - \hz_m(\lambda)) \hPi_m] \ge 0$ holds for any $m \in \mI_M$,
 which yields $L(\Pi, \sigma, \hX, \lambda) \le s(\hX, \lambda)$ from Eq.~(\ref{eq:L2}).
 Thus, from Eq.~(\ref{eq:L_ge_g}), we obtain $L(\Pi, \sigma, \hX, \lambda) = s(\hX, \lambda)$, i.e.,
 \begin{eqnarray}
  \Tr[(\hX - \hz_m(\lambda)) \hPi_m] &=& 0, ~ \forall m \in \mI_M. \label{eq:sum_X_z_Pi_0}
 \end{eqnarray}
 Therefore, using the fact that $\hA\hB = 0$ holds for any $\hA, \hB \in \mS_+$ satisfying $\Tr(\hA\hB) = 0$
 yields Eq.~(\ref{eq:cond_XzPi}).
 From $L(\Pi, \sigma, \hX, \lambda) = f(\Pi)$,
 the fourth term of the right-hand side of Eq.~(\ref{eq:L}) must be zero.
 Therefore, Eq.~(\ref{eq:cond_b}) holds.

 Next, we show (2) $\Rightarrow$ (3).
 From Eq.~(\ref{eq:cond_XzPi}), $\hX \hPi_m = \hz_m(\lambda) \hPi_m$ holds.
 Summing this equation over $m = 0, \cdots, M-1$
 yields $\hX = \sum_{m=0}^{M-1} \hz_m(\lambda) \hPi_m$,
 which gives Eq.~(\ref{eq:cond_Xlambda}).
 Equation~(\ref{eq:cond_b2}) obviously holds from Eq.~(\ref{eq:cond_b}).

 Finally, we show (3) $\Rightarrow$ (1).
 Let $\hX = \sum_{m=0}^{M-1} \hz_m(\lambda) \hPi_m$.
 We have that for any POVM $\Pi' = \{ \hPi'_m : m \in \mI_M \} \in \POVM^\mm$,
 \begin{eqnarray}
  \lefteqn{ f(\Pi) - f(\Pi') } \nonumber \\
  &\ge& f(\Pi) + \sum_{j=0}^{J-1} \lambda_j \left[ b_j - \sum_{m=0}^{M-1} \Tr(\ha_{j,m} \hPi_m) \right] \nonumber \\
  & & \mbox{} - f(\Pi') - \sum_{j=0}^{J-1} \lambda_j \left[ b_j - \sum_{m=0}^{M-1} \Tr(\ha_{j,m} \hPi'_m) \right] \nonumber \\
  &=& \Tr~\hX - \sum_{m=0}^{M-1} \Tr[\hz_m(\lambda) \hPi_m'] \nonumber \\
  &=& \sum_{m=0}^{M-1} \Tr[(\hX - \hz_m(\lambda)) \hPi'_m] \ge 0,
 \end{eqnarray}
 where the first inequality follows from Eq.~(\ref{eq:cond_b2})
 and $\sum_{m=0}^{M-1} \Tr(\ha_{j,m} \hPi'_m) \le b_j$.
 The last inequality follows from Eq.~(\ref{eq:cond_Xlambda}), i.e., $\hX \ge \hz_m(\lambda)$.
 Since $f(\Pi) \ge f(\Pi')$ holds for any POVM $\Pi' \in \POVM^\mm$,
 $\Pi$ is an optimal measurement of problem (\ref{eq:primal}).
 \QED
\end{proof}

\section{Generalized minimax solution} \label{sec:minimax}

\subsection{Formulation}

In this section, we consider the quantum minimax strategy,
which provides a different type of problem from those discussed in the previous section.
The quantum minimax strategy has been investigated
\cite{Hir-Ike-1982,Osa-Ban-Hir-1996,Dar-Sac-Kah-2005,Kat-2012,Nak-Kat-Usu-2013-minimax}
under the assumption that the collection of prior probabilities is not given.
We investigate the minimax strategy for a generalized quantum state discrimination problem.

Let $K$ be a positive integer.
Also, let $\Prob$ be the entire set of collections of $K$ nonnegative real numbers,
$\mu = \{ \mu_k \ge 0 : k \in \mI_K \}$, satisfying $\sum_{k=0}^{K-1} \mu_k = 1$,
which implies that $\mu \in \Prob$ can be interpreted as a probability distribution.
We consider the following function $F(\mu, \Pi)$:
\begin{eqnarray}
 F(\mu, \Pi) &=& \sum_{k=0}^{K-1} \mu_k f_k(\Pi), \nonumber \\
 f_k(\Pi) &=& \sum_{m=0}^{M-1} \Tr(\hc_{k,m} \hPi_m) + d_k, \label{eq:fmm}
\end{eqnarray}
where $\hc_{k,m} \in \mS$ and $d_k \in \Real$ hold for any $m \in \mI_M$ and $k \in \mI_K$.
We want to find a POVM $\Pi \in \POVM^\mm$ that maximizes
the worst-case value of $F(\mu, \Pi)$ over $\mu \in \Prob$, i.e., $\min_{\mu \in \Prob} F(\mu, \Pi)$,
where $\POVM^\mm$ is defined by Eq.~(\ref{eq:POVMmm}).
In the case of $K = 1$, this problem is equivalent to
problem (\ref{eq:primal}) with $\hc_m = \hc_{0,m}$ and $d_0 = 0$.
Therefore, this problem can be regarded as an extension of problem (\ref{eq:primal}).

We can see that if $\POVM^\mm$ is not empty, then the so-called minimax theorem holds, that is,
there exists $(\mu^\star, \Pi^\star)$ satisfying the following equations:
\begin{eqnarray}
 \max_{\Pi \in \POVM^\mm} \min_{\mu \in \Prob} F(\mu, \Pi) &=& F(\mu^\star, \Pi^\star) \nonumber \\
 &=& \min_{\mu \in \Prob} \max_{\Pi \in \POVM^\mm} F(\mu, \Pi). \label{eq:minimax}
\end{eqnarray}
Indeed, $\POVM^\mm$ and $\Prob$ are closed convex sets,
and $F(\mu, \Pi)$ is a continuous convex function of $\mu$ for fixed $\Pi$
and a continuous concave function of $\Pi$ for fixed $\mu$,
which are sufficient conditions for the minimax theorem to hold \cite{Eke-Tem-1999}.
We call $(\mu^\star, \Pi^\star)$, $\mu^\star$, and $\Pi^\star$ respectively a minimax solution,
minimax probabilities, and a minimax measurement.
$(\mu^\star, \Pi^\star)$ is a minimax solution if and only if
$(\mu^\star, \Pi^\star)$ is a saddle point of $F(\mu, \Pi)$,
i.e., the following inequalities hold for any $\mu \in \Prob$ and $\Pi \in \POVM^\mm$ \cite{Eke-Tem-1999}:
\begin{eqnarray}
 F(\mu^\star, \Pi) \le F(\mu^\star, \Pi^\star) \le F(\mu, \Pi^\star). \label{eq:minimax_saddle}
\end{eqnarray}

Let
\begin{eqnarray}
 F^\star(\mu) = \max_{\Pi \in \POVM^\mm} F(\mu, \Pi) \label{eq:Fstar}
\end{eqnarray}
with $\mu \in \Prob$.
It follows from Eq.~(\ref{eq:minimax_saddle}) that $F^\star(\mu^\star) = F(\mu^\star, \Pi^\star)$ holds.
From Eq.~(\ref{eq:fmm}), $F(\mu, \Pi)$ can be expressed by
\begin{eqnarray} \hspace{-1.5em}
 F(\mu, \Pi) &=& \sum_{k=0}^{K-1} \mu_k \left[ \sum_{m=0}^{M-1} \Tr(\hc_{k,m} \hPi_m) + d_k \right] \nonumber \\
 &=& \sum_{m=0}^{M-1} \Tr\left[ \left( \sum_{k=0}^{K-1} \mu_k \hc_{k,m} \right) \hPi_m \right]
  + \sum_{k=0}^{K-1} \mu_k d_k. \label{eq:F_mu_Pi_Fstar}
\end{eqnarray}
Thus, $F^\star(\mu)$ for a given $\mu \in \Prob$ can be obtained
by finding $\Pi \in \POVM^\mm$ that maximizes the first term of the second line of Eq.~(\ref{eq:F_mu_Pi_Fstar}),
which is formulated as problem (\ref{eq:primal}) with $c_m = \sum_{k=0}^{K-1} \mu_k \hc_{k,m}$.

\subsection{Examples} \label{subsec:minimax_example}

We give some examples of minimax problems that can be formulated as Eq.~(\ref{eq:fmm}).
Let us consider discrimination between $R$ quantum states $\{ \hrho_r : r \in \mI_R \}$.

\subsubsection{Minimax solution in the Bayes strategy}

The minimax strategy in which the average Bayes cost is used as the objective function
has been investigated in Ref.~\cite{Kat-2012}.
We regard $\mu \in \Prob$ with $K = R$ as prior probabilities of the states $\{ \hrho_r : r \in \mI_R \}$.
The aim of this problem is to find a POVM $\Pi$ that minimizes the worst-case average Bayes cost
$B(\mu, \Pi)$ over $\mu \in \Prob$.
$B(\mu, \Pi)$ is expressed by
\begin{eqnarray}
 B(\mu, \Pi) &=& \sum_{m=0}^{R-1} \Tr[\hW_m(\mu) \hPi_m], \nonumber \\
 \hW_m(\mu) &=& \sum_{k=0}^{R-1} \mu_k B_{m,k} \hrho_k, \label{eq:minimax_bayes}
\end{eqnarray}
where $B_{m,k} \in \Real_+$ holds for any $m,k \in \mI_R$.
This problem can be expressed by a form of Eq.~(\ref{eq:fmm}) with $F(\mu, \Pi) = - B(\mu, \Pi)$.
In this case, we have
\begin{eqnarray}
 f_k(\Pi) &=& - \sum_{m=0}^{R-1} \Tr[(B_{m,k}\hrho_k) \hPi_m], ~~ \forall k \in \mI_R, \nonumber \\
 \POVM^\mm &=& \POVM,
 \label{eq:bayes_minimax_primal}
\end{eqnarray}
i.e.,
\begin{eqnarray}
 M &=& K = R, \nonumber \\
 J &=& 0, \nonumber \\
 \hc_{k,m} &=& - B_{m,k} \hrho_k, \nonumber \\
 d_k &=& 0.
\end{eqnarray}

\subsubsection{Inconclusive minimax solution}

The application to the minimax strategy to state discrimination
that allows a nonzero failure probability
has been investigated in Ref.~\cite{Nak-Kat-Usu-2013-minimax}.
The aim of this problem is to find a POVM $\Pi$, which we call an inconclusive minimax measurement,
that maximizes the worst-case value of the sum of the average success and failure probabilities
under the constraint that $\Tr(\hrho_j \hPi_R)$ is not greater than a given value $p$
with $0 \le p \le 1$ for any $j \in \mI_R$.
In particular, if $p = 0$, then an inconclusive minimax measurement is
a standard minimax measurement without inconclusive results \cite{Hir-Ike-1982}.
Let $K = R$ and $\mu \in \Prob$ be prior probabilities of the states $\{ \hrho_r : r \in \mI_R \}$;
then, this problem can be expressed by a form of Eq.~(\ref{eq:fmm}) with
\begin{eqnarray}
 f_k(\Pi) &=& \Tr[\hrho_k (\hPi_k + \hPi_R)], ~ \forall k \in \mI_R, \nonumber \\
 \POVM^\mm &=& \{ \Pi \in \POVM : \Tr(\hrho_j \hPi_R) \le p, ~ \forall j \in \mI_R \}.
 \label{eq:inc_minimax_primal}
\end{eqnarray}
That is, we have
\begin{eqnarray}
 M &=& R + 1, \nonumber \\
 K &=& J = R, \nonumber \\
 \hc_{k,m} &=&
  \left\{
   \begin{array}{ll}
	\hrho_k, & ~ m = k ~{\rm or}~ m = R, \\
	0, & ~ {\rm otherwise}, \\
   \end{array} \right. \nonumber \\
 d_k &=& 0, \nonumber \\
 \ha_{j,m} &=& \delta_{m,R} \hrho_j, \nonumber \\
 b_j &=& p.
\end{eqnarray}

\subsubsection{Minimax solution for plural state sets} \label{subsubsec:Pc_minimax}

We consider a quantum measurement that maximizes the worst-case average success probabilities
for plural quantum state sets $\{ \Psi_k : k \in \mI_K \}$ with $K \ge 2$ as another example,
where, for each $k \in \mI_K$,
$\Psi_k$ is a set of $R$ quantum states, $\Psi_k = \{ \hrho_{k,r} : r \in \mI_R \}$,
with prior probabilities $\{ \xi_{k,r} : r \in \mI_R \}$.
This problem can be expressed by a form of Eq.~(\ref{eq:fmm}) with
\begin{eqnarray}
 f_k(\Pi) &=& \sum_{m=0}^{R-1} \Tr(\hrho'_{k,m} \hPi_m), ~~ \forall k \in \mI_K, \nonumber \\
 \POVM^\mm &=& \POVM, \label{eq:Pc_minimax_primal}
\end{eqnarray}
where $\hrho'_{k,r} = \xi_{k,r} \hrho_{k,r}$.
That is, we have
\begin{eqnarray}
 M &=& R, \nonumber \\
 J &=& 0, \nonumber \\
 \hc_{k,m} &=& \hrho'_{k,m}, \nonumber \\
 d_k &=& 0. \label{eq:Pc_minimax_KJ}
\end{eqnarray}
We discuss this problem in detail in Subsec.~\ref{subsec:example_plural}.

\subsection{Properties of a minimax solution}

We show necessary and sufficient conditions for a minimax solution
in Theorem~\ref{thm:minimax}
and an optimization problem of obtaining a minimax measurement
in Theorem~\ref{thm:minimax_convex}.

\begin{thm} \label{thm:minimax}
 Suppose that $\mu^\star \in \Prob$ and $\Pi^\star \in \POVM^\mm$ hold.
 The following statements are all equivalent.
 \begin{enumerate}[(1)]
  \setlength{\parskip}{0cm}
  \setlength{\itemsep}{0cm}
  \item $(\mu^\star, \Pi^\star)$ is a minimax solution to Eq.~(\ref{eq:fmm}).
  \item We have that for any $k \in \mI_K$,
		\begin{eqnarray}
		 f_k(\Pi^\star) &\ge& F^\star(\mu^\star). \label{eq:thm_minimax2}
		\end{eqnarray}
  \item $F^\star(\mu^\star) = F(\mu^\star, \Pi^\star)$ holds,
		and we have that for any $k,k' \in \mI_K$ such that $\mu^\star_{k'} > 0$,
		\begin{eqnarray}
		 f_k(\Pi^\star) &\ge& f_{k'}(\Pi^\star). \label{eq:thm_minimax3}
		\end{eqnarray}
 \end{enumerate}
\end{thm}

\begin{proof}
 It suffices to show (1) $\Leftrightarrow$ (2) and (2) $\Leftrightarrow$ (3).

 First, we show (1) $\Rightarrow$ (2).
 Let $\mu^{(k)} = \{ \mu_{k'} = \delta_{k,k'} : k' \in \mI_K \}$.
 From Eq.~(\ref{eq:minimax_saddle}) and $F^\star(\mu^\star) = F(\mu^\star, \Pi^\star)$,
 we have that for any $k \in \mI_K$,
 \begin{eqnarray}
  f_k(\Pi^\star) &=& F(\mu^{(k)}, \Pi^\star) \ge F(\mu^\star, \Pi^\star) = F^\star(\mu^\star). \label{eq:Fstar_mustar}
 \end{eqnarray}
 Thus, Eq.~(\ref{eq:thm_minimax2}) holds.

 Next, we show (2) $\Rightarrow$ (1).
 From Eqs.~(\ref{eq:Fstar}) and (\ref{eq:thm_minimax2}),
 We obtain, for any $\mu \in \Prob$ and $\Pi \in \POVM^\mm$,
 \begin{eqnarray}
  F(\mu^\star, \Pi) &\le& F^\star(\mu^\star) \le \sum_{k=0}^{K-1} \mu_k f_k(\Pi^\star) = F(\mu, \Pi^\star).
   \label{eq:thm_minimax_FFF}
 \end{eqnarray}
 Substituting $\mu = \mu^\star$ and $\Pi = \Pi^\star$ into this equation
 gives $F^\star(\mu^\star) = F(\mu^\star, \Pi^\star)$.
 Thus, from Eq.~(\ref{eq:thm_minimax_FFF}), Eq.~(\ref{eq:minimax_saddle}) holds,
 which means that $(\mu^\star, \Pi^\star)$ is a minimax solution to Eq.~(\ref{eq:fmm}).

 Then, we show (2) $\Rightarrow$ (3).
 From Eq.~(\ref{eq:thm_minimax2}) and the definition of $F^\star(\mu)$,
 $F(\mu^\star, \Pi^\star) = F^\star(\mu^\star)$ must hold.
 Thus, we have
 \begin{eqnarray}
  f_k(\Pi^\star) &=& F^\star(\mu^\star), ~~ \forall k \in \mI_K ~{\rm s.t.}~\mu^\star_k > 0, \nonumber \\
  f_k(\Pi^\star) &\ge& F^\star(\mu^\star), ~~ \forall k \in \mI_K ~{\rm s.t.}~\mu^\star_k = 0,
   \label{eq:thm_minimax4}
 \end{eqnarray}
 from which we can easily see that Eq.~(\ref{eq:thm_minimax3}) holds.

 Finally, we show (3) $\Rightarrow$ (2).
 From Eq.~(\ref{eq:thm_minimax3}), $f_k(\Pi^\star) = f_{k'}(\Pi^\star)$ holds
 for any $k, k' \in \mI_K$ satisfying $\mu^\star_k > 0$ and $\mu^\star_{k'} > 0$.
 Thus, according to the definition of $F^\star(\mu)$,
 $F^\star(\mu^\star) = f_{k'}(\Pi^\star)$ holds for any $k' \in \mI_K$ satisfying $\mu^\star_{k'} > 0$.
 Substituting this into Eq.~(\ref{eq:thm_minimax3}) gives Eq.~(\ref{eq:thm_minimax2}).
 \QED
\end{proof}

\begin{thm} \label{thm:minimax_convex}
 Let us consider the following optimization problem
 \begin{eqnarray}
  \begin{array}{ll}
   {\rm maximize} & \displaystyle \fmin(\Pi) = \min_{k \in \mI_K} f_k(\Pi) \\
   {\rm subject~to} & \Pi \in \POVM^\mm \\
  \end{array} \label{eq:minimax_convex}
 \end{eqnarray}
 with a POVM $\Pi$.
 A POVM $\Pi^+$ is an optimal solution to problem (\ref{eq:minimax_convex})
 if and only if $\Pi^+$ is a minimax measurement of Eq.~(\ref{eq:fmm}).
\end{thm}

\begin{proof}
 Suppose that $\Pi^+$ is an optimal solution to problem (\ref{eq:minimax_convex}),
 and that $(\mu^\star, \Pi^\star)$ is a minimax solution to Eq.~(\ref{eq:fmm}).
 Equations~(\ref{eq:Fstar}) and (\ref{eq:thm_minimax2}) give
 \begin{eqnarray}
  \fmin(\Pi^\star) &\ge& F^\star(\mu^\star) = \max_{\Pi \in \POVM^\mm} F(\mu^\star, \Pi)
   \ge \max_{\Pi \in \POVM^\mm} \fmin(\Pi), \nonumber \\
  \label{eq:minimax_convex2}
 \end{eqnarray}
 which indicates that $\Pi^\star$ is an optimal solution to problem (\ref{eq:minimax_convex}).
 Since $\Pi^+$ is also an optimal solution to problem (\ref{eq:minimax_convex}),
 $\fmin(\Pi^+) = \fmin(\Pi^\star) \ge F^\star(\mu^\star)$ holds from Eq.~(\ref{eq:minimax_convex2}),
 and thus Statement~(2) of Theorem~\ref{thm:minimax} holds.
 Therefore, $\Pi^+$ is a minimax measurement of Eq.~(\ref{eq:fmm}).
 \QED
\end{proof}

\section{Group covariant optimization problem} \label{sec:sym}

In this section, we show that if an optimization problem of obtaining an optimal measurement
or a minimax solution has a certain symmetry, the optimal solution also has the same symmetry.
A quantum state set that is invariant under the action
of a group $\mG$ in which each element corresponds to a unitary or anti-unitary operator
is called a group covariant (or $\mG$-covariant) state set.
Similarly, we call an optimal measurement and a minimax solution that are invariant under the same action
a group covariant (or $\mG$-covariant) optimal measurement and a minimax solution, respectively.
Optimal measurements for group covariant state sets have been well investigated,
and it has been derived that a $\mG$-covariant optimal measurement exists for
a $\mG$-covariant state set under several optimality criteria
\cite{Bel-1975,Ban-Kur-Mom-Hir-1997,Usu-Tak-Hat-Hir-1999,Eld-For-2001,Eld-2003-inc,Eld-Meg-Ver-2004,Eld-Sto-Has-2004,Nak-Usu-2013-group,Nak-Kat-Usu-2013-minimax}.
These results not only help us to obtain analytical optimal solutions
(e.g., \cite{And-Bar-Gil-Hun-2002,Kat-Hir-2003,Qiu-2008}),
but also are useful for developing computationally efficient algorithms for obtaining optimal solutions
\cite{Ass-Car-Pie-2010,Nak-Kat-Usu-2015-numerical}.
In this section, we will generalize these results to our generalized optimization problems.

\subsection{Group action}

First, let us describe a group action.
A group action of $\mG$ on a set $T$ is a set of mappings from $T$ to $T$,
$\{ \pi_g(x)~(x \in T) : g \in \mG \}$ (we also denote $\pi_g(x)$ as $g \c x$), such that
\begin{enumerate}[(1)]
 \setlength{\parskip}{0cm}
 \setlength{\itemsep}{0cm}
 \item For any $g, h \in \mG$ and $x \in T$, $(gh) \c x = g \c (h \c x)$ holds.
 \item For any $x \in T$, $e \c x = x$ holds, where $e$ is the identity element of $\mG$.
\end{enumerate}
The action of $\mG$ on $T$ is called faithful if, for any distinct $g,h \in \mG$,
there exists $x \in T$ such that $g \c x \neq h \c x$.
Here, we assume that the number of elements in $\mG$, which is denoted as $|\mG|$,
is greater than one.

Let us consider an action of $\mG$ on the set $\mI_N$ with $N \ge 1$, that is,
$\{ g \c n \in \mI_N ~(n \in \mI_N) : g \in \mG \}$.
This action is not faithful in general.
We also consider the action of $\mG$ on $\mS$, expressed by
\begin{eqnarray}
 g \c \hA &=& \hU_g \hA \hU_g^\dagger \label{eq:group_action_S}
\end{eqnarray}
with $g \in \mG$ and $\hA \in \mS$,
where $\hU_g$ is a unitary or anti-unitary operator and $\hU_g^\dagger$ is conjugate transpose of $\hU_g$.
(Note that if $\hU_g$ is an anti-unitary operator, then
$\hU_g^\dagger$ is an anti-unitary operator such that $\hU_g^\dagger \hU_g = \hU_g \hU_g^\dagger = \ident$.)
$\hU_e = \hat{1}$ and $\hU_{\inv{g}} = \hU_g^\dagger$ obviously hold,
where $\inv{g}$ is the inverse element of $g$.
We assume that the action of $\mG$ on $\mS$ is faithful,
which is equivalent to $\hU_g \neq \hU_h$ for any distinct $g, h \in \mG$.
From Eq.~(\ref{eq:group_action_S}),
we can easily verify that for any $g \in \mG$, $c \in \Real$ and $\hA, \hB \in \mS$, we have
\begin{eqnarray}
 g \c (\hA \pm \hB) &=& g \c \hA \pm g \c \hB, \nonumber \\
 g \c (c \hA) &=& c (g \c \hA), \label{eq:g_c} \nonumber \\
 g \c \hat{1} &=& \hat{1}, \label{eq:g_1} \nonumber \\
 \Tr(g \c \hA) &=& \Tr~\hA, \nonumber \\
 \Tr[(g \c \hA)(g \c \hB)] &=& \Tr(\hA\hB), \nonumber \\
 g \c \hA &\in& \mS_+, ~~~~~ \forall \hA \in \mS_+, \nonumber \\
 g \c \hA &\ge& g \c \hB, ~~~ \forall \hA \ge \hB.
\end{eqnarray}
In this section, we use these facts without mentioning them.

\subsection{Group covariant optimal measurement}

As a preparation, we first prove the following lemma.

\begin{lemma} \label{lemma:sym_Pi}
 Suppose that $\POVM^\mm$ is not empty.
 Also, suppose that there exist actions of $\mG$ on $\mS$, $\mI_M$, and $\mI_J$
 such that
 \begin{eqnarray}
  g \c \ha_{j,m} &=& \ha_{g \c j, g \c m}, ~ \forall g \in \mG, j \in \mI_J, m \in \mI_M, \nonumber \\
  b_j &=& b_{g \c j}, ~~~~~~ \forall g \in \mG, j \in \mI_J. \label{eq:sym_constraint}
 \end{eqnarray}
 Let $\kappa_g(\Phi)$ and $\kappa(\Phi)$ be mappings of $\Phi \in \POVM^\mm$ expressed by
 \begin{eqnarray}
  \kappa_g(\Phi) &=& \{ \inv{g} \c \hPhi_{g \c m} : m \in \mI_M \}, \nonumber \\
  \kappa(\Phi) &=& \left\{ \frac{1}{|\mG|} \sum_{g \in \mG} \inv{g} \c \hPhi_{g \c m} : m \in \mI_M \right\}.
   \label{eq:kappa}
 \end{eqnarray}
 Then, $\kappa_g$ is a bijective mapping onto $\POVM^\mm$ for any $g \in \mG$,
 and $\kappa$ is a mapping onto $\POVM^\mm$.
 Moreover, for any $\Phi \in \POVM^\mm$, we have
 \begin{eqnarray}
  g \c \hPi_m &=& \hPi_{g \c m}, ~~ \forall g \in \mG, m \in \mI_M, \label{eq:sym_Pi}
 \end{eqnarray}
 where $\Pi = \kappa(\Phi)$.
\end{lemma}

\begin{proof}
 First, we show that $\kappa_g$ is bijective onto $\POVM^\mm$.
 Let $\Phi \in \POVM^\mm$ and $\Phi^\g = \kappa_g(\Phi)$.
 Since $\Phi^\g_m = \inv{g} \c \hPhi_{g \c m} \in \mS+$
 and $\sum_{m=0}^{M-1} \Phi^\g_m = \inv{g} \c \ident = \ident$ hold,
 $\Phi^\g \in \POVM$ holds.
 We also obtain for any $j \in \mI_J$,
 \begin{eqnarray}
  \sum_{m=0}^{M-1} \Tr(\ha_{j,m} \hPhi^\g_m)
   &=& \sum_{m=0}^{M-1} \Tr[\ha_{j,m} (\inv{g} \c \hPhi_{g \c m})] \nonumber \\
  &=& \sum_{m=0}^{M-1} \Tr[(g \c \ha_{j,m}) \hPhi_{g \c m}] \nonumber \\
  &=& \sum_{m=0}^{M-1} \Tr(\ha_{g \c j,g \c m} \hPhi_{g \c m}) \nonumber \\
  &\le& b_{g \c j} = b_j, \label{eq:sym_Pi_Phig}
 \end{eqnarray}
 where the inequality follows from the group action being bijective and $\Phi \in \POVM^\mm$.
 Thus, $\Phi^\g \in \POVM^\mm$ holds.
 Moreover, since $\kappa_{\inv{g}}[ \kappa_g(\Phi) ] = \kappa_g[ \kappa_{\inv{g}}(\Phi) ] = \Phi$,
 $\kappa_{\inv{g}}$ is the inverse mapping of $\kappa_g$.
 Therefore, $\kappa_g$ is bijective onto $\POVM^\mm$.

 Next, we show that $\kappa$ is a mapping onto $\POVM^\mm$ and that Eq.~(\ref{eq:sym_Pi}) holds.
 From Eq.~(\ref{eq:sym_Pi_Phig}), we have that for any $j \in \mI_J$,
 \begin{eqnarray} \hspace{-1em}
  \sum_{m=0}^{M-1} \Tr(\ha_{j,m} \hPi_m)
   &=& \frac{1}{|G|} \sum_{g \in \mG} \sum_{m=0}^{M-1} \Tr(\ha_{j,m} \hPhi^\g_m) \le b_j,
 \end{eqnarray}
 which means that $\Pi \in \POVM^\mm$ holds for any $\Phi \in \POVM^\mm$,
 that is, $\kappa$ is a mapping onto $\POVM^\mm$.
 We also have that for any $g \in \mG$ and $m \in \mI_M$,
 \begin{eqnarray}
  g \c \hPi_m &=& \frac{1}{|\mG|} \sum_{h \in \mG} g \c \hPhi^{(h)}_m
   = \frac{1}{|\mG|} \sum_{h' \in \mG} \inv{h'} \c \hPhi_{h' \c g \c m} \nonumber \\
  &=& \hPi_{g \c m},
 \end{eqnarray}
 where $h' = h \c \inv{g}$.
 Thus, Eq.~(\ref{eq:sym_Pi}) holds.
 \QED
\end{proof}

We now show that a $\mG$-covariant optimal measurement exists
if optimization problem (\ref{eq:primal}) has a certain symmetry with respect to $\mG$.

\begin{thm} \label{thm:sym}
 Let us consider optimization problem (\ref{eq:primal}).
 Suppose that $\POVM^\mm$ is not empty.
 Also, suppose that there exist actions of $\mG$ on $\mS$, $\mI_M$, and $\mI_J$
 satisfying Eq.~(\ref{eq:sym_constraint}) and
 \begin{eqnarray}
  g \c \hc_m &=& \hc_{g \c m}, ~~ \forall g \in \mG, m \in \mI_M. \label{eq:sym_eval}
 \end{eqnarray}
 Then, for any $\Phi \in \POVM^\mm$
 there exists $\Pi \in \POVM^\mm$ such that $f(\Pi) = f(\Phi)$ and Eq.~(\ref{eq:sym_Pi}) hold,
 where $f$ is the objective function of problem (\ref{eq:primal}).
 In particular, an optimal measurement $\Pi$ exists satisfying Eq.~(\ref{eq:sym_Pi}).
 Moreover, there exists an optimal solution $(\hX, \lambda)$ to dual problem (\ref{eq:dual})
 such that
 \begin{eqnarray}
  g \c \hX &=& \hX, ~~~~ \forall g \in \mG, \nonumber \\
  \lambda_j &=& \lambda_{g \c j}, ~~ \forall g \in \mG, j \in \mI_J. \label{eq:sym_Xlambda}
 \end{eqnarray}
\end{thm}

As examples of Theorem~\ref{thm:sym},
we can derive that there exist a minimum error measurement,
an optimal unambiguous measurement, and an optimal inconclusive measurement
that are $\mG$-covariant if a given state set is $\mG$-covariant,
which is shown in Ref.~\cite{Nak-Usu-2013-group}.

Note that let $\POVM_\mG^\mm$ be the entire set of $\Pi \in \POVM^\mm$ satisfying Eq.~(\ref{eq:sym_Pi});
then, we can easily see that, since $\POVM_\mG^\mm$ is convex,
problem (\ref{eq:primal}) remains in convex programming
even if we restrict the feasible set from $\POVM^\mm$ to $\POVM_\mG^\mm$.

\begin{proof}
 First, we show that $\Pi \in \POVM^\mm$ exists such that
 $f(\Pi) = f(\Phi)$ and Eq.~(\ref{eq:sym_Pi}) hold for any $\Phi \in \POVM^\mm$.
 Let $\Pi = \kappa(\Phi)$, where $\kappa$ is defined by Eq.~(\ref{eq:kappa}).
 From Lemma~\ref{lemma:sym_Pi}, $\Pi$ satisfies $\Pi \in \POVM^\mm$ and Eq.~(\ref{eq:sym_Pi}).
 Moreover, we obtain
 \begin{eqnarray}
  f(\Pi) &=& \sum_{m=0}^{M-1} \Tr(\hc_m \hPi_m) \nonumber \\
  &=& \frac{1}{|\mG|} \sum_{m=0}^{M-1} \sum_{g \in \mG} \Tr[\hc_m (\inv{g} \c \hPhi_{g \c m})] \nonumber \\
  &=& \frac{1}{|\mG|} \sum_{m=0}^{M-1} \sum_{g \in \mG} \Tr[(g \c \hc_m) \hPhi_{g \c m}] \nonumber \\
  &=& \frac{1}{|\mG|} \sum_{g \in \mG} \sum_{m=0}^{M-1} \Tr(\hc_{g \c m} \hPhi_{g \c m}) \nonumber \\
  &=& \frac{1}{|\mG|} \sum_{g \in \mG} f(\Phi) = f(\Phi).
 \end{eqnarray}
 In particular, if $\Phi$ is an optimal measurement, then so is $\Pi$.

 Next, we show that there exists an optimal solution $(\hX, \lambda)$ to dual problem (\ref{eq:dual})
 satisfying Eq.~(\ref{eq:sym_Xlambda}).
 Let $\nu = \{ \nu_j : j \in \mI_J \} \in \Real_+^J$.
 Suppose that $(\hY, \nu)$ is an optimal solution to problem (\ref{eq:dual}).
 Also, let $\hY^\g = g \c \hY$ and $\nu^\g = \{ \nu_j^\g = \nu_{\inv{g} \c j} : j \in \mI_J \}$.
 $\hY^\g \in \mS$ and $\nu^\g \in \Real_+^J$ obviously hold.
 We obtain for any $g \in \mG$ and $m \in \mI_M$,
 \begin{eqnarray}
  \hY^\g &\ge& g \c \hz_m(\nu) = \hc_{g \c m} - \sum_{j=0}^{J-1} \nu_j \ha_{g \c j,g \c m} \nonumber \\
  &=& \hc_{g \c m} - \sum_{j=0}^{J-1} \nu_{g \c j}^\g \ha_{g \c j,g \c m} = \hz_{g \c m}(\nu^\g).
   \label{eq:sym_Xl1}
 \end{eqnarray}
 We also obtain
 \begin{eqnarray}
  s(\hY^\g, \nu^\g) &=& \Tr~\hY^\g + \sum_{j=0}^{J-1} \nu^\g_j b_j \nonumber \\
  &=& \Tr~\hY + \sum_{j=0}^{J-1} \nu_{\inv{g} \c j} b_{\inv{g} \c j} = s(\hY, \nu). \label{eq:sym_Xl2}
 \end{eqnarray}
 From Eqs.~(\ref{eq:sym_Xl1}) and (\ref{eq:sym_Xl2}),
 $(\hY^\g, \nu^\g)$ is also an optimal solution to problem (\ref{eq:dual}).
 Let $\hX \in \mS$ and $\lambda = \{ \lambda_j : j \in \mI_J \} \in \Real_+^J$ be expressed by
 \begin{eqnarray}
  \hX &=& \frac{1}{|\mG|} \sum_{g \in \mG} \hY^\g, ~~~
   \lambda_j = \frac{1}{|\mG|} \sum_{g \in \mG} \nu_j^\g. \label{eq:sym_Xlambda2}
 \end{eqnarray}
 We can easily see that Eq.~(\ref{eq:sym_Xlambda}) holds.
 For any $m \in \mI_M$, we have
 \begin{eqnarray}
  \hz_m(\lambda) &=& \hc_m - \frac{1}{|\mG|} \sum_{g \in \mG} \sum_{j=0}^{J-1} \nu_j^\g \ha_{j,m} \nonumber \\
  &=& \frac{1}{|\mG|} \sum_{g \in \mG} \left( \hc_m - \sum_{j=0}^{J-1} \nu_j^\g \ha_{j,m} \right) \nonumber \\
  &=& \frac{1}{|\mG|} \sum_{g \in \mG} \hz_m(\nu^\g). \label{eq:sym_z_lambda}
 \end{eqnarray}
 From Eqs.~(\ref{eq:sym_Xl1}), (\ref{eq:sym_Xlambda2}), and (\ref{eq:sym_z_lambda}), we obtain for any $m \in \mI_M$,
 \begin{eqnarray}
  \hX - \hz_m(\lambda) &=& \frac{1}{|\mG|} \sum_{g \in \mG} [\hY^\g - \hz_m(\nu^\g)] \ge 0.
 \end{eqnarray}
 Moreover, from Eqs.~(\ref{eq:sym_Xl2}) and (\ref{eq:sym_Xlambda2}), we have
 \begin{eqnarray}
  s(\hX, \lambda) &=& \Tr~\hX + \sum_{j=0}^{J-1} \lambda_j b_j \nonumber \\
  &=& \frac{1}{|\mG|} \sum_{g \in \mG} (\Tr~\hY^{(b)} + \nu_j^\g b_j) = s(\hY, \nu).
 \end{eqnarray}
 Therefore, $(\hX, \nu)$ is also an optimal solution to problem (\ref{eq:dual}).
 \QED
\end{proof}

\subsection{Group covariant minimax solution}

Similar to Theorem~\ref{thm:sym},
we can show that if Eq.~(\ref{eq:fmm}) has a certain symmetry with respect to $\mG$,
then there exists a $\mG$-covariant minimax solution.

\begin{thm} \label{thm:sym_minimax}
 Let us consider a minimax solution to Eq.~(\ref{eq:fmm}).
 Suppose that $\POVM^\mm$ is not empty.
 Also, suppose that there exist actions of $\mG$ on $\mS$, $\mI_M$, $\mI_J$, and $\mI_K$
 satisfying Eq.~(\ref{eq:sym_constraint}) and
 \begin{eqnarray}
  g \c \hc_{k,m} &=& \hc_{g \c k,g \c m}, ~~ \forall g \in \mG, k \in \mI_K, m \in \mI_M, \nonumber \\
  d_k &=& d_{g \c k}, ~~~~~~~ \forall g \in \mG, k \in \mI_K. \label{eq:sym_minimax_cd}
 \end{eqnarray}
 Then, a minimax solution $(\mu, \Pi)$ exists such that
 \begin{eqnarray}
  \mu_k &=& \mu_{g \c k}, ~~~ \forall g \in \mG, k \in \mI_K, \nonumber \\
  g \c \hPi_m &=& \hPi_{g \c m}, ~~ \forall g \in \mG, m \in \mI_M. \label{eq:sym_minimax}
 \end{eqnarray}
\end{thm}

\begin{proof}
 Let $(\eta^\star, \Pi^\star)$ be a minimax solution to Eq.~(\ref{eq:fmm}).
 Also, let $\mu = \{ \mu_k = |\mG|^{-1} \sum_{g \in \mG} \eta^\star_{g \c k} : k \in \mI_K \}$
 and $\Pi = \kappa(\Pi^\star)$,
 where $\kappa$ is defined by Eq.~(\ref{eq:kappa}).
 Then, it follows that $\mu \in \Prob$, $\Pi \in \POVM^\mm$, and Eq.~(\ref{eq:sym_minimax}) hold
 (also see Lemma~\ref{lemma:sym_Pi}).
 Here, we show that $(\mu, \Pi)$ is a minimax solution to Eq.~(\ref{eq:fmm}).
 From Statement~(2) of Theorem~\ref{thm:minimax},
 it suffices to show that $f_k(\Pi) \ge F^\star(\mu)$ holds for any $k \in \mI_K$.
 We will show $f_k(\Pi) \ge F^\star(\eta^\star)$ and $F^\star(\eta^\star) \ge F^\star(\mu)$.

 First, we show $f_k(\Pi) \ge F^\star(\eta^\star)$ for any $k \in \mI_K$.
 Let $\Pi^\g = \kappa_g(\Pi^\star)$; then for any $k \in \mI_K$, we have
 \begin{eqnarray}
  f_k(\Pi) &=& \frac{1}{|\mG|} \sum_{m=0}^{M-1} \sum_{g \in \mG} \Tr(\hc_{k,m} \hPi^\g_m) + d_k \nonumber \\
  &=& \frac{1}{|\mG|} \sum_{g \in \mG} \left[ \sum_{m=0}^{M-1} \Tr(\hc_{k,m} \hPi^\g_m) + d_k \right] \nonumber \\
  &=& \frac{1}{|\mG|} \sum_{g \in \mG} \left[ \sum_{m=0}^{M-1} \Tr[(g \c \hc_{k,m}) \hPi^\star_{g \c m}] + d_k \right]
   \nonumber \\
  &=& \frac{1}{|\mG|} \sum_{g \in \mG} \left[ \sum_{m'=0}^{M-1} \Tr(\hc_{g \c k,m'} \hPi^\star_{m'}) + d_k \right]
   \nonumber \\
  &=& \frac{1}{|\mG|} \sum_{g \in \mG} f_{g \c k}(\Pi^\star) \ge F^\star(\eta^\star),
 \end{eqnarray}
 where $m' = g \c m$.
 The inequality in the last line follows from $f_k(\Pi^\star) \ge F^\star(\eta^\star)$ for any $k \in \mI_K$,
 which is obtained from Theorem~\ref{thm:minimax}.

 Next, we show $F^\star(\eta^\star) \ge F^\star(\mu)$.
 Let $\eta^\g = \{ \eta^\star_{g \c k} : k \in \mI_K \}$.
 We have that for any $g \in \mG$,
 \begin{eqnarray}
  F^\star(\eta^\g) &=& \max_{\Phi \in \POVM^\mm} \sum_{k=0}^{K-1} \eta^\star_{g \c k}
   \left[ \sum_{m=0}^{M-1} \Tr(\hc_{k,m}\hPhi_m) + d_k \right] \nonumber \\
  &=& \max_{\Phi \in \POVM^\mm} \sum_{k'=0}^{K-1} \eta^\star_{k'}
   \left[ \sum_{m=0}^{M-1} \Tr[\hc_{k',m'}(g \c \hPhi_m)] + d_{k'} \right] \nonumber \\
  &=& \max_{\Phi' \in \POVM^\mm} \sum_{k'=0}^{K-1} \eta^\star_{k'}
   \left[ \sum_{m'=0}^{M-1} \Tr(\hc_{k',m'}\hPhi'_{m'}) + d_{k'} \right] \nonumber \\
  &=& F^\star(\eta^\star), \label{eq:sym_minimax_Fstar_eta}
 \end{eqnarray}
 where $k' = g \c k$, $m' = g \c m$, and $\Phi' = \kappa_{\inv{g}}(\Phi)$.
 The third line follows from the mapping $\kappa_{\inv{g}}$ being bijective onto $\POVM^\mm$
 (see Lemma~\ref{lemma:sym_Pi}).
 From Eq.~(\ref{eq:sym_minimax_Fstar_eta}), we obtain
 \begin{eqnarray}
  F^\star(\mu) &=& \max_{\Phi \in \POVM^\mm} \frac{1}{|\mG|} \sum_{g \in \mG} \sum_{k=0}^{K-1} \eta^\g_k f_k(\Phi)
   \nonumber \\
  &\le& \frac{1}{|\mG|} \sum_{g \in \mG} F^\star(\eta^\g) = F^\star(\eta^\star).
 \end{eqnarray}
 Therefore, $(\mu, \Pi)$ is a minimax solution.
 \QED
\end{proof}

\section{Examples of optimal measurement and minimax solution} \label{sec:example}

As an example of a generalized optimal measurement,
we discuss the problem of finding an optimal inconclusive measurement
with a lower bound on success probabilities, which is introduced in Subsubsec.~\ref{subsubsec:bounded_inc}.
Also, as an example of a generalized minimax solution,
we discuss the problem of finding a minimax solution for plural state sets,
which is introduced in Subsubsec.~\ref{subsubsec:Pc_minimax}.
Moreover, Tables~\ref{tab:example} and \ref{tab:minimax_example} summarize the problem formulations
and their examples shown in Subsecs.~\ref{subsec:example} and \ref{subsec:minimax_example},
respectively.

\begin{table*}[tbp]
 \caption{Basic formulation of generalized optimal measurements and its examples.}
 \label{tab:example}
 \begin{center}
  \tabcolsep = 0.4em
  \begin{tabular}{p{0.1em}lll}
   \hline \hline
   & \multicolumn{1}{c}{Primal problems} & \multicolumn{1}{c}{Dual problems} & \multicolumn{1}{c}{Necessary and sufficient conditions} \\
   & & & \multicolumn{1}{c}{(Statement~(3) of Theorem~\ref{thm:condition})} \\ \hline
   \multicolumn{4}{l}{Basic formulation} \\
   & \raisebox{-0.0em}{\shortstack[l]{
     ${\rm maximize} ~ \displaystyle \sum_{m=0}^{M-1} \Tr(\hc_m \hPi_m)$ \\
     ${\rm subject~to} ~ \Pi \in \POVM$, \\
     \hspace{1em} $\displaystyle \sum_{m=0}^{M-1} \Tr(\ha_{j,m} \hPi_m) \le b_j, ~ \forall j \in \mI_J$ \\
     \hfill (\ref{eq:primal}),(\ref{eq:POVMmm})}} &
   \raisebox{-0.4em}{\shortstack[l]{
     ${\rm minimize} ~~ \displaystyle \Tr~\hX + \sum_{j=0}^{J-1} \lambda_j b_j$ \\
     ${\rm subject~to} ~ \hX \ge \hz_m(\lambda), ~ \forall m \in \mI_M$ \\
     where ~ $\displaystyle \hz_m(\lambda) = \hc_m - \sum_{j=0}^{J-1} \lambda_j \ha_{j,m}$ \\
     \hfill (\ref{eq:dual}),(\ref{eq:zm})}} &
   \raisebox{-2.4em}{\shortstack[l]{
     $\lambda \in \Real_+^J$ exists such that \\
     \hspace{0.5em} $\hX(\lambda) \ge \hz_m(\lambda), ~ \forall m \in \mI_M$, \\
     \hspace{0.5em} $\displaystyle \lambda_j \left[ b_j - \sum_{m=0}^{M-1} \Tr(\ha_{j,m} \hPi_m) \right] = 0,
                    ~ \forall j \in \mI_J$ \\
     where ~ $\displaystyle \hX(\lambda) = \sum_{n=0}^{M-1} \hz_n(\lambda) \hPi_n$ \\
     \hfill (\ref{eq:cond_Xz}),(\ref{eq:cond_XzPi}),(\ref{eq:cond_b})}} \\
   \multicolumn{4}{l}{Example~1: Optimal measurement in the Bayes criterion (Subsec.~\ref{subsubsec:bayes})
     \cite{Hol-1973,Yue-Ken-Lax-1975,Hel-1976}} \\
   & \raisebox{-2.1em}{\shortstack[l]{
	 ${\rm minimize} ~~ \displaystyle \sum_{m=0}^{R-1} \Tr(\hW_m \hPi_m)$ \\
     ${\rm subject~to} ~ \Pi \in \POVM$
     \hfill (\ref{eq:bayes_primal})}} &
   \raisebox{-1.0em}{\shortstack[l]{
	 ${\rm maximize} ~ \displaystyle \Tr~\hX$ \\
     $\displaystyle {\rm subject~to} ~ \hW_m \ge \hX, ~ \forall m \in \mI_R$}} &
   \raisebox{-0.0em}{\shortstack[l]{
     $\displaystyle \hW_m \ge \sum_{r=0}^{R-1} \hW_r \hPi_r, ~ \forall m \in \mI_R$}} \\
   \multicolumn{4}{l}{Example~2: Optimal error margin measurement (Subsec.~\ref{subsubsec:error_margin})
     \cite{Tou-Ada-Ste-2007,Hay-Has-Hor-2008,Sug-Has-Hor-Hay-2009}} \\
   & \raisebox{-3.8em}{\shortstack[l]{
     ${\rm maximize} ~ \displaystyle \sum_{r=0}^{R-1} \xi_r \Tr(\hrho_r \hPi_r)$ \\
     ${\rm subject~to} ~ \Pi \in \POVM$, \\
     \hspace{1em} $\displaystyle \sum_{r=0}^{R-1} \xi_r \Tr[\hrho_r (\hPi_r + \hPi_R)] \ge 1 - \varepsilon$ \\
     \hfill (\ref{eq:err_primal})}} &
   \raisebox{-0.0em}{\shortstack[l]{
	 ${\rm minimize} ~~ \Tr~\hX - \lambda(1 - \varepsilon)$ \\
     ${\rm subject~to} ~ \hX \ge (1 + \lambda) \xi_r \hrho_r, ~ \forall r \in \mI_R$, \\
     \hspace{1em} $\hX \ge \lambda \hG$}} &
   \raisebox{-5.6em}{\shortstack[l]{
     $\lambda \in \Real_+$ exists such that \\
     \hspace{0.5em} $\hX(\lambda) \ge (1 + \lambda) \xi_r \hrho_r, ~ \forall r \in \mI_R$, \\
     \hspace{0.5em} $\hX(\lambda) \ge \lambda \hG$, \\
     \hspace{0.5em} $\displaystyle \lambda \left[ \sum_{r=0}^{R-1} \xi_r \Tr[\hrho_r (\hPi_r + \hPi_R)]
                   - 1 + \varepsilon \right] = 0$ \\
     where ~ $\displaystyle \hX(\lambda) = (1 + \lambda) \sum_{r=0}^{R-1} \xi_r \hrho_r \hPi_r
              + \lambda \hG \hPi_R$}} \\
   \multicolumn{4}{l}{Example~3: Optimal inconclusive measurement with a lower bound on success probabilities (Subsec.~\ref{subsubsec:bounded_inc})} \\
   & \raisebox{-1.2em}{\shortstack[l]{
	 ${\rm maximize} ~ \displaystyle \sum_{r=0}^{R-1} \xi_r \Tr(\hrho_r \hPi_r)$ \\
     ${\rm subject~to} ~ \Pi \in \POVM$, \\
     \hspace{1em} $\Tr(\hrho_r \hPi_r) \ge q, ~ \forall r \in \mI_R$, \\
     \hspace{1em} $\Tr(\hG\hPi_R) \ge p$
     \hfill (\ref{eq:bounded_inc_primal})}} &
   \raisebox{-0.0em}{\shortstack[l]{
	 ${\rm minimize} ~~ \displaystyle \Tr~\hX - q \sum_{r=0}^{R-1} \lambda_r - p \lambda_R$ \\
     $\displaystyle {\rm subject~to} ~ \hX \ge (\xi_r + \lambda_r) \hrho_r, ~ \forall r \in \mI_R$, \\
     \hspace{1em} $\hX \ge \lambda_R \hG$
     \hfill (\ref{eq:bounded_inc_dual})}} &
   \raisebox{-7.7em}{\shortstack[l]{
     $\lambda \in \Real_+^{R+1}$ exists such that \\
     \hspace{0.5em} $\hX \ge (\xi_r + \lambda_r) \hrho_r, ~ \forall r \in \mI_R$, \\
     \hspace{0.5em} $\hX \ge \lambda_R \hG$, \\
     \hspace{0.5em} $\lambda_r \left[ \Tr(\hrho_r \hPi_r) - q \right] = 0, ~ \forall r \in \mI_R$, \\
     \hspace{0.5em} $\lambda_R \left[ \Tr(\hG\hPi_R) - p \right] = 0$ \\
     where ~ $\displaystyle \hX(\lambda) = \sum_{r=0}^{R-1} (\xi_r + \lambda_r) \hrho_r \hPi_r
              + \lambda_R \hG \hPi_R$ \\
     \hfill (\ref{eq:bounded_inc_cond}),(\ref{eq:bounded_inc_X})}} \\
   \hline \hline
  \end{tabular}
 \end{center}
\end{table*}

\begin{table*}[tbp]
 \caption{Basic formulation of generalized minimax solutions and its examples.}
 \label{tab:minimax_example}
 \begin{center}
  \tabcolsep = 0.5em
  \begin{tabular}{p{0.1em}ll}
   \hline \hline
   & \multicolumn{1}{c}{Problems} & \multicolumn{1}{c}{Necessary and sufficient conditions} \\
   & & \multicolumn{1}{c}{(Statement~(3) of Theorem~\ref{thm:minimax})} \\ \hline
   \multicolumn{3}{l}{Basic formulation} \\
   & \raisebox{-11.8em}{\shortstack[l]{
     $\displaystyle {\rm maximize} \min_{u \in \Prob} F(\mu, \Pi)$ \\
     ${\rm subject~to} ~ \Pi \in \POVM^\mm$ \\
     where ~ $\displaystyle F(\mu, \Pi) = \sum_{k=0}^{K-1} \mu_k f_k(\Pi)$, \\
     \hspace{1em} ~~~ $\displaystyle f_k(\Pi) = \sum_{m=0}^{M-1} \Tr(\hc_{k,m} \hPi_m) + d_k$, \\
     \hspace{1em} ~~~ $\displaystyle \POVM^\mm = \left\{ \Pi \in \POVM : \sum_{m=0}^{M-1} \Tr(\ha_{j,m} \hPi_m)
                       \le b_j, ~ \forall j \in \mI_J \right\}$ \\
     \hfill (\ref{eq:POVMmm}),(\ref{eq:fmm}),(\ref{eq:minimax})}} &
   \raisebox{-0.0em}{\shortstack[l]{
	 $F^\star(\mu^\star) = F(\mu^\star, \Pi^\star)$, \\
     $f_k(\Pi^\star) \ge f_{k'}(\Pi^\star), ~ \forall k,k' \in \mI_K ~{\rm s.t.}~\mu^\star_{k'} > 0$
      ~~ (\ref{eq:thm_minimax3})}} \\
   \multicolumn{3}{l}{Example~1: Minimax solution in Bayes strategy \cite{Kat-2012}} \\
   & \raisebox{-0.0em}{\shortstack[l]{
     $\displaystyle {\rm minimize} \max_{u \in \Prob} \sum_{k=0}^{R-1} \mu_k \sum_{m=0}^{R-1}
	  B_{m,k} \Tr(\hrho_k \hPi_m)$ \\
     ${\rm subject~to} ~ \Pi \in \POVM$
     \hfill (\ref{eq:minimax_bayes})}} &
   \raisebox{-2.5em}{\shortstack[l]{
	 $F^\star(\mu^\star) = F(\mu^\star, \Pi^\star)$, \\
     $\displaystyle \sum_{m=0}^{R-1} B_{m,k} \Tr(\hrho_k \hPi^\star_m) \le
      \sum_{m=0}^{R-1} B_{m,k'} \Tr(\hrho_{k'} \hPi^\star_m)$, \\
     \hspace{1em} $\forall k,k' \in \mI_R ~{\rm s.t.}~\mu^\star_{k'} > 0$}} \\
   \multicolumn{3}{l}{Example~2: Inconclusive minimax solution \cite{Nak-Kat-Usu-2013-minimax}} \\
   & \raisebox{-1.0em}{\shortstack[l]{
     $\displaystyle {\rm maximize} \min_{u \in \Prob} \sum_{k=0}^{R-1} \mu_k \Tr[\hrho_k (\hPi_k + \hPi_R)]$ \\
     ${\rm subject~to} ~ \Pi \in \POVM^\mm$ \\
     where ~ $\displaystyle \POVM^\mm = \left\{ \Pi \in \POVM : \Tr(\hrho_j \hPi_R) \le p,
              ~ \forall j \in \mI_R \right\}$
     \hfill (\ref{eq:inc_minimax_primal})}} &
   \raisebox{-0.0em}{\shortstack[l]{
	 $F^\star(\mu^\star) = F(\mu^\star, \Pi^\star)$, \\
     $\displaystyle \Tr[\hrho_k (\hPi^\star_k + \hPi^\star_R)] \ge
      \Tr[\hrho_{k'} (\hPi^\star_{k'} + \hPi^\star_R)]$, \\
     \hspace{1em} $\forall k,k' \in \mI_R ~{\rm s.t.}~\mu^\star_{k'} > 0$}} \\
   \multicolumn{3}{l}{Example~3: Minimax solution for plural state sets} \\
   & \raisebox{-0.0em}{\shortstack[l]{
     $\displaystyle {\rm maximize} \min_{u \in \Prob} \sum_{k=0}^{K-1} \mu_k \sum_{m=0}^{R-1}
      \Tr(\hrho'_{k,m} \hPi_m)$ \\
     ${\rm subject~to} ~ \Pi \in \POVM$
     \hfill (\ref{eq:Pc_minimax_primal})}} &
   \raisebox{-2.5em}{\shortstack[l]{
	 $F^\star(\mu^\star) = F(\mu^\star, \Pi^\star)$, \\
     $\displaystyle \sum_{m=0}^{R-1} \Tr(\hrho'_{k,m} \hPi^\star_m) \ge
      \sum_{m=0}^{R-1} \Tr(\hrho'_{k',m} \hPi^\star_m)$, \\
     \hspace{1em} $\forall k,k' \in \mI_K ~{\rm s.t.}~\mu^\star_{k'} > 0$}} ~ (\ref{eq:minimax_plural}) \\
   \hline \hline
  \end{tabular}
 \end{center}
\end{table*}

\subsection{Optimal inconclusive measurement with a lower bound on success probabilities}
 \label{subsec:example_bounded_inc}

In this example, we can apply Theorems~\ref{thm:dual} and \ref{thm:condition}.
Substituting Eq.~(\ref{eq:bounded_inc_abc}) into Eq.~(\ref{eq:zm}) gives
\begin{eqnarray}
 \hz_m(\lambda) &=&
  \left\{
   \begin{array}{ll}
	(\xi_m + \lambda_m) \hrho_m, & ~ m < R, \\
	\lambda_R \hG, & ~ m = R.
   \end{array} \right.
\end{eqnarray}
Thus, from Theorem~\ref{thm:dual}, dual problem (\ref{eq:dual}) can be rewritten as
\begin{eqnarray}
 \begin{array}{ll}
  {\rm minimize} & \displaystyle s(\hX, \lambda) = \Tr~\hX - q \sum_{r=0}^{R-1} \lambda_r - p \lambda_R \\
  {\rm subject~to} & \hX \ge (\xi_r + \lambda_r) \hrho_r, ~~ \forall r \in \mI_R,  \\
  & \hX \ge \lambda_R \hG. \label{eq:bounded_inc_dual}
 \end{array}
\end{eqnarray}
$\lambda_R \hG \in \mS_+$ yields $\hX \in \mS_+$.
In particular, in the case of $q = 0$, Eq.~(\ref{eq:bounded_inc_dual}) is equivalent to
the dual problem of finding an optimal inconclusive measurement,
which is shown in Theorem~1 of Ref.~\cite{Eld-2003-inc}.

We can also obtain necessary and sufficient conditions for an optimal measurement
from Theorem~\ref{thm:condition}.
For example, from Statement~(3) of this theorem,
$\Pi \in \POVM^\mm$ is an optimal measurement of problem (\ref{eq:bounded_inc_primal})
if and only if $\lambda \in \Real_+^J$ exists such that
\begin{eqnarray}
 \hX(\lambda) - (\xi_r + \lambda_r) \hrho_r &\ge& 0, ~~ \forall r \in \mI_R, \nonumber \\
 \hX(\lambda) - \lambda_R \hG &\ge& 0, \nonumber \\
 \lambda_r \left[ \Tr(\hrho_r \hPi_r) - q \right] &=& 0, ~~~ \forall r \in \mI_R, \nonumber \\
 \lambda_R \left[ \Tr(\hG\hPi_R) - p \right] &=& 0, \label{eq:bounded_inc_cond}
\end{eqnarray}
where
\begin{eqnarray}
 \hX(\lambda) &=& \sum_{r=0}^{R-1} (\xi_r + \lambda_r) \hrho_r \hPi_r + \lambda_R \hG \hPi_R.
  \label{eq:bounded_inc_X}
\end{eqnarray}

In the case in which problem (\ref{eq:bounded_inc_primal}) has a certain symmetry,
we can apply Theorem~\ref{thm:sym}.
Suppose that a given state set is $\mG$-covariant, that is,
there exist actions of $\mG$ on $\mS$ and $\mI_R$ satisfying
$g \c (\xi_r \hrho_r) = \xi_{g \c r} \hrho_{g \c r}$,
which is equivalent to $g \c \hrho_r = \hrho_{g \c r}$ and $\xi_r = \xi_{g \c r}$,
for any $g \in \mG$ and $r \in \mI_R$.
Let the action of $\mG$ on $\mI_M = \mI_{R+1}$,
$g \c m$ $~(m \in \mI_M)$, be $g \c m = \pi^{(R)}_g(m)$ for any $m \in \mI_R$ and $g \c R = R$,
where $\{ \pi^{(R)}_g : g \in \mG \}$ is the action of $\mG$ on $\mI_R$.
Also, let the action of $\mG$ on $\mI_J$ be the same as the action of $\mG$ on $\mI_M$.
Then, since Eqs.~(\ref{eq:sym_constraint}) and (\ref{eq:sym_eval}) hold,
there exists an optimal measurement satisfying Eq.~(\ref{eq:sym_Pi}).

\subsection{Minimax solution for plural state sets} \label{subsec:example_plural}

In this example, we can apply Theorem~\ref{thm:minimax},
that is, $(\mu^\star, \Pi^\star)$ is a minimax solution if and only if
Eq.~(\ref{eq:thm_minimax2}) holds, or $F^\star(\mu^\star) = F(\mu^\star, \Pi^\star)$ and
Eq.~(\ref{eq:thm_minimax3}) hold.
Substituting Eq.~(\ref{eq:Pc_minimax_primal}) into Eq.~(\ref{eq:thm_minimax3}) gives
\begin{eqnarray}
 && \sum_{m=0}^{R-1} \Tr(\hrho'_{k,m} \hPi^\star_m) \ge \sum_{m=0}^{R-1} \Tr(\hrho'_{k',m} \hPi^\star_m), \nonumber \\
 && ~~~ \forall k,k' \in \mI_K ~{\rm s.t.}~\mu^\star_{k'} > 0.
  \label{eq:minimax_plural}
\end{eqnarray}

From Eq.~(\ref{eq:Pc_minimax_KJ}), $F(\mu,\Pi)$ is expressed by
\begin{eqnarray}
 F(\mu,\Pi) &=& \sum_{k=0}^{K-1} \mu_k \sum_{m=0}^{R-1} \Tr(\hrho'_{k,m} \hPi_m) \nonumber \\
 &=& \sum_{m=0}^{R-1} \Tr\left[ \left( \sum_{k=0}^{K-1} \mu_k \hrho'_{k,m} \right) \hPi_m \right].
\end{eqnarray}
Thus, $F^\star(\mu)$ is equivalent to the optimal value of $f(\Pi)$
of optimization problem (\ref{eq:primal}) with
\begin{eqnarray}
 M &=& R, \nonumber \\
 J &=& 0, \nonumber \\
 \hc_m &=& \sum_{k=0}^{K-1} \mu_k \hrho'_{k,m}. \label{eq:minimax_plural_Fstar_param}
\end{eqnarray}
This indicates that $F^\star(\mu)$ is also equivalent to the average success probability
of a minimum error measurement for the state set $\{ \hc_m / \Tr~\hc_m : m \in \mI_R \}$
with prior probabilities $\{ \Tr~\hc_m : m \in \mI_R\}$
(note that $\hc_m \in \mS_+$ holds from Eq.~(\ref{eq:minimax_plural_Fstar_param})).

We can also apply Theorem~\ref{thm:sym_minimax} in the case in which
given state sets have a certain symmetry.
Assume that there exist actions of $\mG$ on $\mS$, $\mI_R$, and $\mI_K$
satisfying $g \c \hrho'_{k,m} = \hrho'_{g \c k,g \c m}$
for any $g \in \mG$, $m \in \mI_R$, and $k \in \mI_K$.
For example, this assumption holds if each state set $\Psi_k = \{ \hrho_{k,m} : m \in \mI_R \}$
is $\mG$-covariant under the same actions of $\mG$ on $\mS$ and $\mI_R$,
i.e., $g \c \hrho'_{k,m} = \hrho'_{k,g \c m}$ for any $g \in \mG$ and $m \in \mI_R$
(in this case let $g \c k = k$ for any $k \in \mI_K$).
Under this assumption, Eq.~(\ref{eq:sym_minimax_cd}) holds,
and thus a $\mG$-covariant minimax solution exists.

\section{Conclusion}

We investigated a generalized optimization problem of finding quantum measurements.
Each of the objective and constraint functions in this problem is formulated by
the sum of the traces of the multiplication of a Hermitian operator and a detection operator.
We first derived corresponding dual problems and necessary and sufficient conditions
for an optimal measurement.
The minimax version of this problem was also studied,
and necessary and sufficient conditions for a minimax solution were provided.
We finally showed that for an optimization problem having a certain symmetry
with respect to a group in which each element corresponds to a unitary or anti-unitary operator,
there exists an optimal solution with the same symmetry.

\begin{acknowledgments}
 We are grateful to O. Hirota of Tamagawa University for support.
 T. S. U. was supported (in part) by JSPS KAKENHI (Grant No.24360151).
\end{acknowledgments}

\appendix

\section{Derivation of dual problem without Born rule} \label{append:nosig}

Here, we show that the optimal value of problem (\ref{eq:dual}) is an upper bound of
the objective function $f(\Pi)$ for a generalized optimal measurement $\Pi$
without using the Born rule (i.e., the fact that
the probability $P(m|\hrho)$ of the outcome $m \in \mI_M$ for input state $\hrho$
is $\Tr(\hrho\hPi_m)$).
We pose the following two requirements:
\begin{enumerate}[(1)]
 \setlength{\parskip}{0cm}
 \setlength{\itemsep}{0cm}
 \item Any quantum state is given by a density operator,
	   which is positive semidefinite with unit trace.
 \item The probability $P(m|\hrho)$ is affine in $\hrho$,
	   that is, for any two states $\hrho$ and $\hrho'$ and any $t \in \Real$ with $0 \le t \le 1$,
	   we have
	   \begin{eqnarray}
		P[m|t\hrho + (1-t)\hrho'] &=& t P(m|\hrho) + (1-t)P(m|\hrho'). \label{eq:nosig_Prob}
	   \end{eqnarray}
\end{enumerate}
Since $t\hrho + (1-t)\hrho'$ can be interpreted as a statistical mixture of the states
$\hrho$ and $\hrho'$ with probabilities $t$ and $1-t$,
the second requirement seems to be natural,
which is also pointed out in Ref.~\cite{Kim-Miy-Ima-2009}.
Since $P(m|\hrho)$ is a probability, it must satisfy
$P(m|\hrho) \ge 0$ for any $m \in \mI_M$ and $\sum_{m=0}^{M-1} P(m|\hrho) = 1$.

To simplify the notation,
we extend $P(m|\hrho)$ to a linear mapping, which we denote as $p(m|\hrho)$, as follows.
$p(m|\hA)$ $~(m \in \mI_M, \hA \in \mS)$ is defined such that
$p(m|\hrho) = P(m|\hrho)$ holds for any density operator $\hrho$
and it satisfies, for any $t, t' \in \Real$ and $\hA, \hA' \in \mS$,
\begin{eqnarray}
 p(m|t \hA + t' \hA') &=& t p(m|\hA) + t' p(m|\hA'). \label{eq:nosig_P}
\end{eqnarray}
This equation means that $p(m|\hA)$ is linear in $\hA$.
This definition uniquely determines $p(m|\cdot)$ for a given $P(m|\cdot)$.
Since any $\hA \in \mS_+$ can be expressed by a form of $\hA = t \hrho$ with $t = \Tr~\hA \ge 0$
and a density operator $\hrho = \hA / \Tr~\hA$, we obtain
\begin{eqnarray}
 p(m|\hA) \ge 0, ~~ \forall \hA \in \mS_+. \label{eq:nosig_p_ge_0}
\end{eqnarray}
Moreover, for any $\hA \in \mS$, let a Schmidt decomposition of $\hA$
be $\hA = \sum_n \lambda_n \hP_n$
($\hP_n$ can be regarded as a density operator);
then, from the linearity of $p(m|\cdot)$ and $\sum_{m=0}^{M-1} P(m|\hP_n) = 1$,
we obtain
\begin{eqnarray}
 \sum_{m=0}^{M-1} p(m|\hA) &=& \sum_n \lambda_n \sum_{m=0}^{M-1} P(m|\hP_n) = \Tr~\hA. \label{eq:nosig_sum1}
\end{eqnarray}

It should be noted that we can derive from the above two requirements (and Eq.~(\ref{eq:nosig_P})) that,
for any quantum measurement, there exists a POVM $\Pi = \{ \hPi_m : m \in \mI_M \}$
satisfying $P(m|\hrho) = \Tr(\hrho\hPi_m)$, i.e., the Born rule holds;
however, we do not use this fact in this section.

To avoid using the Born rule, we consider the following problem instead of problem (\ref{eq:primal}):
\begin{eqnarray}
 \begin{array}{ll}
  {\rm maximize} & \displaystyle f_p(\Pi) = \sum_{m=0}^{M-1} p(m|\hc_m) \\
  {\rm subject~to} & \Pi \in \POVM^\bullet, \\
 \end{array} \label{eq:nosig_primal}
\end{eqnarray}
where $\hc_m \in \mS$ holds for any $m \in \mI_M$.
$\Pi$ is a quantum measurement, which is expressed as a collection of mappings $p(m|\cdot)$,
i.e., $\{ p(m|\cdot) : m \in \mI_M \}$.
$\POVM^\bullet$ is defined by
\begin{eqnarray}
 \hspace{-1em}
  \POVM^\bullet &=& \left\{ \Pi : \sum_{m=0}^{M-1} p(m|\ha_{j,m}) \le b_j, ~ \forall j \in \mI_J \right\},
  \label{eq:POVMbullet}
\end{eqnarray}
where $\ha_{j,m} \in \mS$ and $b_j \in \Real$ hold for any $m \in \mI_M$ and $j \in \mI_J$.
$J$ is a nonnegative integer.

Let $\lambda \in \Real_+^J$.
Also, choose $\hX$ such that $\hX \ge \hz_m(\lambda)$ holds for any $m \in \mI_M$,
where $\hz_m(\lambda)$ is defined by Eq.~(\ref{eq:zm}).
From Eq.~(\ref{eq:nosig_p_ge_0}), for any $m \in \mI_M$, we have
\begin{eqnarray}
 p(m|\hX) - p[m|\hz_m(\lambda)] &=& p[m|\hX - \hz_m(\lambda)] \ge 0.
\end{eqnarray}
Thus, from Eq.~(\ref{eq:nosig_sum1}), we have
\begin{eqnarray}
 \sum_{m=0}^{M-1} p[m|\hz_m(\lambda)] &\le& \sum_{m=0}^{M-1} p(m|\hX) = \Tr~\hX. \label{eq:nosig_ineq}
\end{eqnarray}
Therefore, we obtain for any $\Pi \in \POVM^\bullet$,
\begin{eqnarray}
 f_p(\Pi)
 &\le& \sum_{m=0}^{M-1} p(m|\hc_m) + \sum_{j=0}^{J-1} \lambda_j
 \left[ b_j - \sum_{m=0}^{M-1} p(m|\ha_{j,m}) \right] \nonumber \\
 &=& \sum_{m=0}^{M-1} p[m|\hz_m(\lambda)] + \sum_{j=0}^{J-1} \lambda_j b_j
 \le \Tr~\hX + \sum_{j=0}^{J-1} \lambda_j b_j. \nonumber \\
 \label{eq:nosig_proof}
\end{eqnarray}
where the equality in the second line follows from Eq.~(\ref{eq:zm}).
Equation~(\ref{eq:nosig_proof}) means that
the optimal value of problem (\ref{eq:dual}) provides
an upper bound of the optimal value of problem (\ref{eq:nosig_primal}).

It is worth mentioning that the discussion given above has a strong relationship with
the approach described in Refs.~\cite{Bae-Hwa-Han-2011,Bae-2013},
in which it is pointed out that
the average success probability of a minimum error measurement is upper bounded
by ensemble steering and the no-signaling principle,
and its upper bound equals to the average success probability.
In preparation, we introduce ensemble steering.
Assume that two parties, Alice and Bob, share an entangled state and the reduced state
on Bob's side is $\hrho$ that can represent
\begin{eqnarray}
 \hrho = \sum_{n=0}^{N-1} q_n \hrho_n \label{eq:nosig_steering}
\end{eqnarray}
with $N \ge 2$,
where $\hrho_n$ is a density operator and $q_n \ge 0$ satisfies $\sum_{n=0}^{N-1} q_n = 1$.
Then, there exists an Alice's measurement with $N$ outcomes that
prepares Bob's state $\hrho_n$ with probability $q_n$.
This is known as ensemble steering,
which was first noted by Schr{\"o}dinger \cite{Sch-1935,Sch-1936}
and also formalized as the Gisin-Hughston-Jozsa-Wootters theorem \cite{Gis-1989,Hug-Joz-Woo-1993}.
The probability that Bob obtains the result $m$ given that his state is $\hrho$ can be expressed
\begin{eqnarray}
 P(m|\hrho) &=& \sum_{n=0}^{N-1} q_n P(m|\hrho_n), \label{eq:nosig_steering_prob}
\end{eqnarray}
where the right-hand side denotes the weighted average of the conditional probabilities
that Bob obtains the result $m$ knowing that his state are $\hrho_0, \cdots, \hrho_{N-1}$
with weights $q_0, \cdots, q_{N-1}$.
Equations~(\ref{eq:nosig_steering}) and (\ref{eq:nosig_steering_prob}) mean that
$P(m|\hrho)$ is affine in $\hrho$.
In other words, the property that $P(m|\hrho)$ is affine in $\hrho$ can be derived from ensemble steering;
thus, one can apply the discussion described in this section.
Note that the approach described in Refs.~\cite{Bae-Hwa-Han-2011,Bae-2013} also provides
an operational interpretation for the average success probability of a minimum error measurement.

%

\end{document}